\documentclass[journal]{IEEEtran}
\usepackage{graphicx,indentfirst}
\usepackage{booktabs}
\usepackage{indentfirst}
\usepackage{amsmath}
\usepackage{cite}
\usepackage{stfloats}
\usepackage{multicol}
\usepackage{makecell}
\usepackage{multirow}
\usepackage{slashbox}
\usepackage{subfigure}
\usepackage[english]{babel}
\usepackage{amsthm}
\usepackage{bm}
\usepackage{amsmath}
\usepackage{multirow}
\usepackage{paralist}
\usepackage{algorithm, algorithmic}
\usepackage{amssymb}
\usepackage{amssymb}
\usepackage{mathrsfs}
\usepackage{amsfonts}
\usepackage{color}

\newtheorem{theorem}{Theorem}
\newtheorem{lemma}{Lemma}

\theoremstyle{plain}
\newtheorem{definition}{Definition}

\IEEEoverridecommandlockouts
\begin{document}
\bibliographystyle{IEEE2}
\title{Cloud/Fog Computing Resource Management and Pricing for Blockchain Networks}
\author{Zehui~Xiong, Shaohan~Feng, Wenbo Wang, Dusit~Niyato,~\IEEEmembership{Fellow,~IEEE,} \\Ping~Wang,~\IEEEmembership{Senior Member,~IEEE,} and Zhu Han, \IEEEmembership{Fellow,~IEEE}\\
\thanks{An earlier version of this paper was accepted by IEEE ICC in~\cite{xiong2017optimal}.}
\thanks{Z. Xiong, S. Feng, W. Wang, and D. Niyato are with School of Computer Science and Engineering, Nanyang Technological University, Singapore. P. Wang is with Department of Electrical Engineering and Computer Science, York University, Canada. Z. Han is with the University of Houston, Houston, USA, and also with the Department of Computer Science and Engineering, Kyung Hee University, Seoul, South Korea.}
\thanks{Copyright (c) 2012 IEEE. Personal use of this material is permitted. However, permission to use this material for any other purposes must be obtained from the IEEE by sending a request to pubs-permissions@ieee.org.}}
\maketitle
\begin{abstract}
Public blockchain networks using proof of work-based consensus protocols are considered a promising platform for decentralized resource management with financial incentive mechanisms. In order to maintain a secured, universal state of the blockchain, proof of work-based consensus protocols financially incentivize the nodes in the network to compete for the privilege of block generation through cryptographic puzzle solving. For rational consensus nodes, i.e., miners with limited local computational resources, offloading the computation load for proof of work to the cloud/fog providers becomes a viable option. In this paper, we study the interaction between the cloud/fog providers and the miners in a proof of work-based blockchain network using a game theoretic approach. In particular, we propose a lightweight infrastructure of the proof of work-based blockchains, where the computation-intensive part of the consensus process is offloaded to the cloud/fog. We formulate the computation resource management in the blockchain consensus process as a two-stage Stackelberg game, where the profit of the cloud/fog provider and the utilities of the individual miners are jointly optimized. In the first stage of the game, the cloud/fog provider sets the price of offered computing resource. In the second stage, the miners decide on the amount of service to purchase accordingly. We apply backward induction to analyze the sub-game perfect equilibria in each stage for both uniform and discriminatory pricing schemes. For uniform pricing where the same price applies to all miners, the uniqueness of the Stackelberg equilibrium is validated by identifying the best response strategies of the miners. For discriminatory pricing where the different prices are applied, the uniqueness of the Stackelberg equilibrium is proved by capitalizing on the variational inequality theory. Further, the real experimental results are employed to justify our proposed model.
\end{abstract}

\section{Introduction}\label{Sec:Introduction}

Blockchain networks were first designed to be the backbone of a distributed, permissionless/public database for recording the transactional data of cryptocurrencies in a tamper-proof and totally ordered manner~\cite{Bitcoin, wang2018survey}. The blockchain network is essentially organized as a virtual overlay Peer-to-Peer (P2P) network, where the database state is maintained in a purely decentralized manner and any node in the network is allowed to join the state maintenance process without the need of identity authentication. As indicated by the name ``blockchain'', the records of transactions between nodes in the network are organized in a data structure known as the ``block''. A series of blocks are arranged in a strictly increasing-time order by a linked-list-like data structure known as the chain of blocks (i.e., ``blockchain''). The blockchain is maintained as the appending-only local replicas by the nodes participating in the replicated consensus process. Unlike the traditional distributed ledger systems using the Practical Byzantine Faulty-Tolerant (PBFT)~\cite{castro1999practical} or Paxos~\cite{ongaro2014search} protocols, a permissionless blockchain network no longer needs any centralized authorities (e.g., authenticating/authorizing servers) and is able to accommodate a much larger number of consensus nodes in the network~\cite{Vukolic2016}. Such an objective is achieved by blockchain networks with the Nakamoto consensus protocol~\cite{Bitcoin} (or protocols alike). Per the Nakamoto protocol, financial incentive is introduced into the consensus process to ensure that the best strategies of the pseudonymous consensus nodes is to follow the given rules of blockchain maintenance/extension. Otherwise they will suffer from monetary loss.

The core component of the Nakamoto consensus protocol is a computation-intensive process known as Proof of Work (PoW). For the consensus nodes that propose their local blockchain view to be the new state of the blockchain database, PoW requires them to solve a cryptographic puzzle, i.e., to find a partial preimage satisfying certain conditions of a hash mapping based on the proposed blockchain state. According to~\cite{Garay2015}, a typical PoW process is executed in the following steps. First, with an input contribution function, a consensus node validates and bundles a sub-set of unconfirmed transactions into a new block. Then, the consensus node computes the PoW solution to the cryptographic puzzle, which is formed based on the value of the new block. Immediately after the puzzle solution is obtained, the consensus node broadcasts the new block to the entire network as its own proposal of the new blockchain head. On the other hand, the rest of nodes in the network run a chain validation-comparison function to determine whether to accept such a proposal or not. In the blockchain network, an honest consensus node follows ``the-longest-chain'' rule and adopts the longest one among the received blockain proposals to update its local view of the blockchain state. In such a process, the nodes that devote their computational resources to the generation of new blocks (i.e., PoW solutions) are also known as the block ``miners''. This is mainly because according to the Nakamoto protocol, a certain amount of blockchain tokens will be awarded to the node that has its proposed blockchain state accepted by the majority of the network. The theoretic proof and analysis for secure and private communication with the Nakamoto protocol can be found in~\cite{Garay2015}.

With the blossom of various cryptocurrencies, permissionless blockchains are considered to be especially appropriate for constructing the decentralized autonomous resource management framework in (wireless) communication networks. Specifically, when the resource management relies on the design of incentive mechanisms (e.g., resource access control~\cite{maesa2017blockchain} and proactive edge caching~\cite{Wang1805:Decentralized}), permissionless blockchains are able to provide fast implementation of the self-organized trading platform with small investment in the operational infrastructure. Furthermore, with the PoW-based Nakamoto consensus protocol, the users of a Decentralized Application (DApp) are incentivized to turn themselves from the free riders of the blockchain network into consensus nodes (i.e., block miners) for more profit. However, due to the required computation contribution by the PoW, the computationally lightweight nodes such as the Internet of Things (IoT) devices may be prevented from directly participating in the consensus process. To alleviate such limitation, ``cloud mining'' becomes a viable option where the mobile devices offload their storage load and/or computation tasks in PoW to the Cloud/Fog Providers (CFPs) or even other edge devices~\cite{chen2017framework, huang2017vehicular}. In the case of computation offloading, the lightweight devices may employ the existing cloud-mining protocols such as Stratum~\cite{recabarren2017hardening} without causing any significant transmission overhead. From the perspective of the blockchain-based DApp's designer, the benefit of encouraging cloud-based mining is multifold. First, by incorporating more consensus nodes, the robustness of the blockchain network is naturally improved~\cite{Garay2015}. Second, the user devices may improve their valuation of the DApps, thanks to the additional reward obtained in the consensus process. Also, the high level of user activities may attract more users and in return further improve the robustness of the underlying blockchain network.

In this paper, we study the interaction between the computationally lightweight devices and a CFP, where the lightweight devices (i.e., block miners) purchase the computing power from the CFP to participate in the consensus process of a PoW-based blockchain for block-mining revenues. Game theory can be leveraged as a promising mathematical tool to analyze the interactions among the CFP and block miners. For example, in~\cite{zhang2017hierarchical}, the authors formulated a Stackelberg game to solve the resource management in fog computing networks, where the game theoretic study of the market and pricing strategies are presented. In~\cite{zhang2017multi}, the authors studied the spectrum resource allocation in order to mitigate the interference management among multiple cellular operators in the unlicensed system. A multi-leader multi-follower Stackelberg game is proposed to model the interactions among the operators and users in unlicensed spectrum. Similarly, we also model the resource offloading market as a two-stage Stackelberg game. In the first stage, the CFP sets the unit price for computation offloading. In the second stage, the miners decide on the amount of services to purchase from the CFP. In particular, we analyze two pricing schemes~\cite{jiang2014optimal}, i.e., uniform pricing where a uniform unit price is applied to all the miners and discriminatory pricing where different unit prices are assigned to different miners. The uniform pricing leads to a straightforward implementation as the CFP does not need to keep track of information of every miner, and charging the same prices is fair to all miners. However, from the perspective of the CFP, discriminatory pricing yields a higher profit by allowing price adjustment for different miners~\cite{laffont1998network}. The main contributions of this paper are summarized as follows.

\begin{enumerate}
\item We explore the possibility of implementing a permissionless, PoW-based blockchain in a network of computationally lightweight devices. By allowing computation offloading to the cloud/fog, we model the interactions between the rational blockchain miners and the CFP as a two-stage Stackelberg game.

\item We study both the uniform pricing scheme and the discriminatory pricing scheme for the CFP. Through backward induction, we provide a series of analytically results with respect to the properties of the Stackelberg equilibrium in different situations.

\item In particular, the existence and uniqueness of Stackelberg equilibrium are validated by identifying the best response strategies of the miners under the uniform pricing scheme. Likewise, the Stackelberg equilibrium is proved to exist and be unique by capitalizing on the Variational Inequalities (VI) theory under discriminatory pricing scheme.

\item We conduct extensive numerical simulations to evaluate the performance of the proposed price-based resource management in blockchain networks. The results show that the discriminatory pricing helps the CFP to encourage more service demand from the miners and achieve greater profit. Moreover, under uniform pricing, the CFP has an incentive to set the maximum price for the profit maximization.
\end{enumerate}

The rest of the paper is organized as follows. Section~\ref{Sec:RelatedWork} presents a brief review of the related work. We describe the model of the consensus formation in a permissionless PoW-based blockchain network and formulate the two-stage Stackelberg game between the lightweight nodes and the CFP in Section~\ref{Sec:Model}. In Section~\ref{Sec:Solution}, we analyze the optimal service demand of block miners as well as the profit maximization of the CFP using backward induction for both uniform and discriminatory pricing schemes. We present the performance evaluations in Section~\ref{Sec:Simulation}. Section~\ref{Sec:Conclusion} concludes the paper with summary and future directions.

\section{Related Work}\label{Sec:RelatedWork}
\subsection{Public Blockchains, DApps and Incentive Mechanism}
For blockchain networks, the core technological ``building blocks'' have been recognized as the distributed database (i.e., ledger), the consensus protocol and the executable scripts (i.e., smart contract) based on network consensus~\cite{dinh2017untangling}. From a data processing point of view, a DApp is essentially a collection of smart contracts and transactional data residing on the blockchain. The realization of a DApp relies on the distributed ledger to identify the state/ownership changes of the tokenized assets. The smart contracts are implemented as transaction (data)-driven procedures to autonomously determine the state transition regarding the asset re-distribution among the DApp users~\cite{dinh2017untangling}. With public blockchains, the implementation of a DApp does not require a centralized infrastructure, namely, dedicated storage and computation provision for the ledger and smart contracts. Instead, the DApp users are allowed to freely enable their functionalities among transaction issuing/validation, information propagation/storage and consensus participation~\cite{tschorsch2016bitcoin, dinh2017untangling}. More specifically, the token-based incentive mechanisms in public blockchains offload the tasks of resource provision and system maintenance from the DApp providers to the DApp users. Thereby, public blockhain networks are considered to be a suitable platform for implementing the incentive-driven Distributed Autonomous Organization (DAO) systems.

In recent years, a line of work has been dedicated to the study in DAO for wireless networking applications based on public blockchains. In~\cite{8024034}, a trading platform for Device-to-Device (D2D) computation offloading is proposed using a dedicated cryptocurrency network. Therein, resource offloading is executed between neighbor D2D nodes through smart contract-based auctions, and the block mining tasks are offloaded to the cloudlets. In~\cite{7966965}, a PoW-based public blockchain is adopted as the backbone of a P2P file storage market, where the privacy of different parties in a transaction is enhanced by the techniques such as ring signatures and one-time payment
addresses. When identity verification is required for market access granting, e.g., in the senarios of autonomous network slice brokering~\cite{8260929} and P2P electricity trading~\cite{kang2017enabling}, the public blockchains can be adapted into consortium blockchains by introducing membership authorizing servers with little modification to the consensus protocols and smart contract design.

Our paper also relates to the classical literature on incentive mechanisms in crowdsensing~\cite{yang2012crowdsourcing, chakeri2017incentive, chakeri2017iterative}. In crowdsensing, the crowdsensing platform as the service provider offers a reward as the incentive to attract more crowdsensing user participation. In the pioneering work~\cite{yang2012crowdsourcing}, the authors considered two system models: the platform-centric model where the provider offers a certain amount of reward that will be shared by the participating users, and the user-centric model where the users have their reserve prices for the participation. In~\cite{chakeri2017incentive}, the authors designed the incentive mechanisms for crowdsensing with multiple crowdsourcers, i.e., service providers. The interactions among the service providers are modelled as the noncooperative game. Therein, the authors proposed a discrete time dynamic algorithm utilizing the best response dynamics to compute the Nash equilibrium of the modeled game. The authors in~\cite{chakeri2017iterative} presented the incentive mechanism in a sealed market where the users have incomplete information on other users' behavior. The convergence to the Nash equilibrium in such a market is then analyzed using the well-known best response dynamics.

\subsection{Consensus and Game Theoretic Mining Models in PoW-based Blockchains}
By the Nakamoto protocol, from a single miner's point of view, the process of solving a PoW puzzle involves an exhaustive query to a collision-resistant hash function (e.g., SHA-256), which aims to find a fixed-length hashcode output with no less than a given number of prefix zeros~\cite{Bitcoin,Garay2015}. For each individual miner, such a process simulates a Poisson process when the required number of prefix zeros is sufficiently large. For a group of miners independently running their own PoW processes at the same time, the first miner to obtain the PoW puzzle solution will have a high probability of getting its block head proposal acknowledged by the entire network. Therefore, block mining under the Nakamoto protocol can also be viewed as a hashing competition, where the probability of a miner winning the competition is roughly proportional to the ratio between its devoted hash power\footnote{We use the hash power and computing power interchangeably throughout the paper.} and the total hash power in the network.

According to the theoretical analysis in~\cite{Garay2015}, when the PoW-based blockchain network satisfies the condition of honest majority in terms of computing power, the probability for the blockchain state machine to be compromised is negligible. Therefore, the mainstream research on the PoW-based consensus protocols focus on the protocol's incentive compatibility and thus the search of miners' rational strategy to optimize the reward obtained in the mining process. A plethora of recent studies~\cite{houy2014bitcoin,beccuti2017bitcoin,kiayias2016blockchain} model the mining process in PoW-based blockchain networks as a noncooperative game, where rational miners may withhold their newly found blocks with valid PoW solutions to internationally cause the fork of the blockchain. In certain conditions of hash power distribution, it is proved in~\cite{houy2014bitcoin,beccuti2017bitcoin,kiayias2016blockchain} that by postponing the newly mined blocks, rational miners may obtain a higher expected payoff than fully abiding by the Nakamoto protocol.

In the literature, the most relevant works to this paper are about the pool-based mining mechanisms. In public blockchains based on outsourceable PoW schemes, a mining pool is essentially a proxy node in the network that only enables its local functionalities of transaction issuing/validation and information propagation/storage. The proxy node offloads the queries to the hash function to the mining workers that subscribe to the pool for mining payment~\cite{dinh2017untangling,tschorsch2016bitcoin}. It is worth noting that most of the existing studies consider the pool-based mining from the perspective of mining workers (i.e., cloud-side resource providers)~\cite{lewenberg2015bitcoin,fisch2017socially, kim2016group, luu2015power}.  In~\cite{lewenberg2015bitcoin}, the process of mining pool formation is modeled as a coalitional game among the mining workers, which is found to have an empty core under the proportional payment scheme. In contrast, the social welfare of miners is considered in~\cite{fisch2017socially} and a geometric-payment pooling strategy is found to be able to achieve the optimal steady-state utility for the miners. In~\cite{kim2016group}, the group bargaining solution is adopted by considering the P2P relationship of the miners. In~\cite{luu2015power}, instead of limiting the miner subscription to a single mining pool, a computing power-splitting game is proposed. With the proposed scheme, the miners play a puzzle-solution game by distributing their computing power into different pools in order to maximize the mining reward.

\section{System Model and Game Formulation}\label{Sec:Model}

In this section, we first propose the system model of blockchain under our consideration~\cite{xiong2018mobile}. Then, we present the Stackelberg game formulation for the price-based computing resource management in blockchain networks assisted by cloud/fog computing.

\begin{figure}[t]
\centering
\includegraphics[width=.4\textwidth]{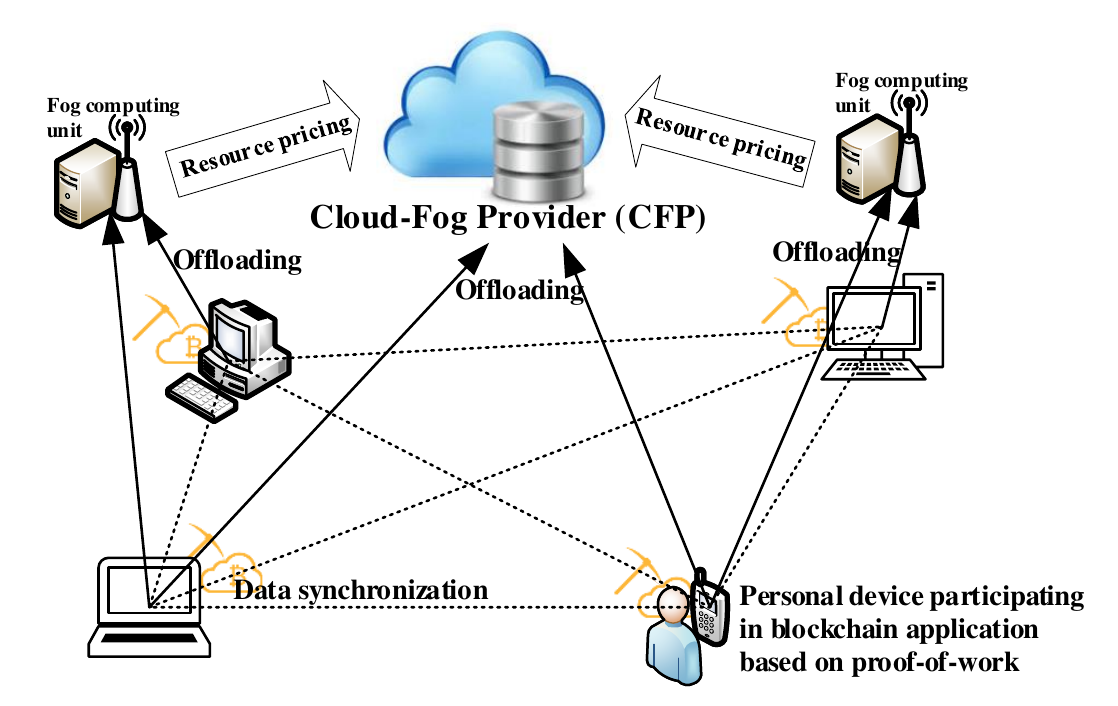}
\caption{\footnotesize{System model of public blockchain application involving PoW.}}\label{Fig:Model}
\end{figure}

\subsection{Chain Mining Assisted by Cloud/Fog Computing}\label{Subsec:ChainMining}
We consider a public blockchain network using the PoW-based consensus protocol~\cite{xiong2017optimal, luong2017blockchain, jiao2017blockchain}. The blockchain network dedicatedly works as the backbone of a specific DApp, where most of the nodes are limited in their local computing power (e.g., the IoT devices and smart phones in a typical crowd-sensing market). We assume that the adopted PoW protocol is ASIC-resistant~\cite{tschorsch2016bitcoin}, e.g., using the Ethash-based PoW scheme~\cite{wood2014ethereum} or the schemes alike. Then, to participate in the consensus process, a node only has to solve the PoW puzzle with general-purpose computing devices. In the blockchain network, a set of $N$ nodes denoted as ${\mathcal{N}} = \{1, \ldots, N\}$, are interested in participating in the consensus process and make extra profit through block mining. In order to achieve this, these block miners purchase the necessary hash power from a public CFP (e.g., Amazon EC2) without hassle of managing the infrastructure such as seeking extra electricity sources~\cite{tosh2017security}. In addition, we consider that the CFP is able to provide the near-to-end computing units such as fog nodes or even edge devices which are closer to the miners\footnote{Note that this fog unit deployment is also more appropriate in hostile environment where the communications with remote cloud are limited and for the access from personal devices which keep moving, e.g., mobile devices.}~\cite{NanTSC}. As such, the aforementioned PoW puzzle can be offloaded to the remote cloud or the nearby fog computing unit. The computing resources offered to the miners is priced by the CFP\footnote{Note that the resource may also include communication resource. Specifically, we can consider that the communication cost is part of the price charged by the CFP. In other words, the CFP offers the service as a bundle which is composed of computing and wireless/wired communication resources. The energy consumption for the computing and communication is naturally accounted in the bundle.}. Figure~\ref{Fig:Model} shows the system model of the blockchain network under our consideration. Note that we assume that the link between the miners and cloud/fog computing units is sufficiently reliable and secured, which is guaranteed by certain ready-to-use communication protocols (e.g., Stratum~\cite{recabarren2017hardening}).

The CFP, i.e., the seller, sells the computing services, and the miners, i.e., the buyers, access and consume this service from the remote cloud or the nearby fog computing unit. Each miner~$i \in {\cal N}$ determines their individual service demand, denoted by $x_i$. Additionally, we consider $x_i \in [\underline x ,\overline x]$, in which $\underline x$ is the minimum service demand, e.g., for blockchain data synchronization, and $\overline x$ is the maximum service demand governed by the CFP. Note that each miner has no incentive to unboundedly increase its service demand due to its financial burden. Then, let $\mathbf{x} \buildrel \Delta \over = ({x_1}, \ldots ,{x_N})$ and ${\mathbf{x}}_{-i}$ represent the service demand profile of all the miners and all other miners except miner~$i$, respectively. As such, the miner $i \in \cal N$ with the service demand $x_i$ has a relative computing power (hash power) $\alpha_i$ with respect to the total hash power of the network, which is defined as follows:
\begin{equation}\label{Eq:HashingPower}
{\alpha_i}(x_i, {\bf x}_{-i}) = \frac{{{x_i}}}{{\sum\nolimits_{j\in \cal N} {{x_j}} }}, \alpha_i >0,
\end{equation}
such that ${\sum\nolimits_{j\in \cal N} {{\alpha_j}} }=1$.

In the blockchain network, miners compete against each other in order to be the first one to solve the PoW puzzle and receive the reward from the speed game accordingly. The occurrence of solving the puzzle can be modeled as a random variable following a Poisson process with mean $\lambda = \frac{1}{{600\sec}}$~\cite{houy2014bitcoin}. Note that our model is general that can be applied with other values of $\lambda$ easily.
The set of transactions to be included in a block chosen by miner $i$ is denoted as $t_i$. Once the miner successfully solves the puzzle, the miner needs to propagate its solution to the whole blockchain network and its solution needs to reach consensus. Because there is no centralized authority to verify the validate a newly mined block, a mechanism for reaching network consensus must be employed. In this mechanism, the verification needs to be processed by other miners before the new mined block is appended to the current blockchain.

The first miner to successfully mine a block that reaches consensus earns the reward. The reward consists of a fixed reward denoted by $R$, and a variable reward which is defined as $r t_i$, where $r$ denotes a given variable reward factor and $t_i$ denotes the number of transactions included in the block mined by miner $i$~\cite{houy2014bitcoin}. Additionally, the process of solving the puzzle incurs an associated cost, i.e., the payment from miner $i$ to the CFP, $p_i$. The objective of the miners is to maximize their individual expected utility, and for miner $i$, it is defined as follows:
\begin{equation}
u_i = (R + r{t_i} )P_i\left(\alpha_i({x_i},{{\bf{x}}_{ - i}}), t_i\right) - {p_i}{x_i},
\end{equation}
where $P\left(\alpha_i ({x_i},{{\bf{x}}_{ - i}}), t_i\right)$ is the probability that miner $i$ successfully mines the block and its solutions reach consensus, i.e., miner $i$ wins the mining reward.

The process of successfully mining a block consists of two steps, i.e., the mining step and the propagation step. In the mining step, the probability that miner $i$ mines the block is directly proportional to its relative computing power $\alpha_i$. Furthermore, there are diminishing chances of wining if one miner chooses to propagate a block that propagates slowly to other miners in the propagation step. In other words, even though one miner may find the first valid block, if its mined block is large, then this block will be likely to be discarded because of long latency, which is called orphaning~\cite{houy2014bitcoin}. Considering this fact, the probability of successful mining by miner $i$ is discounted by the chances that the block is orphaned, ${\mathbb P}_{\mathrm{orphan}}(t_i)$, which is expressed by
\begin{equation}
P_i(\alpha_i ( x_i, {\mathbf{x}}_{-i} ), t_i ) = {\alpha_i}(1 - {{\mathbb P}_{\mathrm{orphan}}}(t_i)).
\end{equation}
Using the fact that block mining times follow the Poisson distribution aforementioned, the orphaning probability is approximated as~\cite{Approximation}:
\begin{equation}
{\mathbb P}_{\mathrm{orphan}}(t_i) = 1 - {e^{ - \lambda \tau(t_i) }},
\end{equation}
where $\tau(t_i)$ is the block propagation time, which is a function of the block size. In other words, the propagation time needed for a block to reach consensus is dependent on its size $t_i$, i.e., the number of transactions in it~\cite{houy2014bitcoin, senmarti2015analysis}. Thus, the bigger the block is, the more time needed to propagate the block to the whole blockchain network~\cite{decker2013information}. Same as~\cite{houy2014bitcoin}, we assume this time function is linear, i.e., $\tau(t_i) = z \times t_i$ with $z>0$ represents a given delay factor. Note that this linear approximation is acceptable according to the numerical results from~\cite{houy2014bitcoin}. Additionally, it would be more appropriate to add a constant term in this function~\cite{decker2013information}, but apparently this constant term has no effect on our subsequent analytical results. Thus, the probability that the miner $i$ successfully mines a block and its solution reaches consensus is expressed as follows:
\begin{equation}
{P_i}(\alpha_i ({x_i},{{\bf{x}}_{ - i}}),t_i) = {\alpha_i}{e^{ -\lambda z{t_i}}},
\end{equation}
where ${\alpha_i}(x_i, {\bf x}_{-i})$ is given in~(\ref{Eq:HashingPower}).

\subsection{Two-Stage Stackelberg Game Formulation}\label{Subsec:GameModel}
\begin{figure}[t]
\centering
\includegraphics[width=0.4\textwidth]{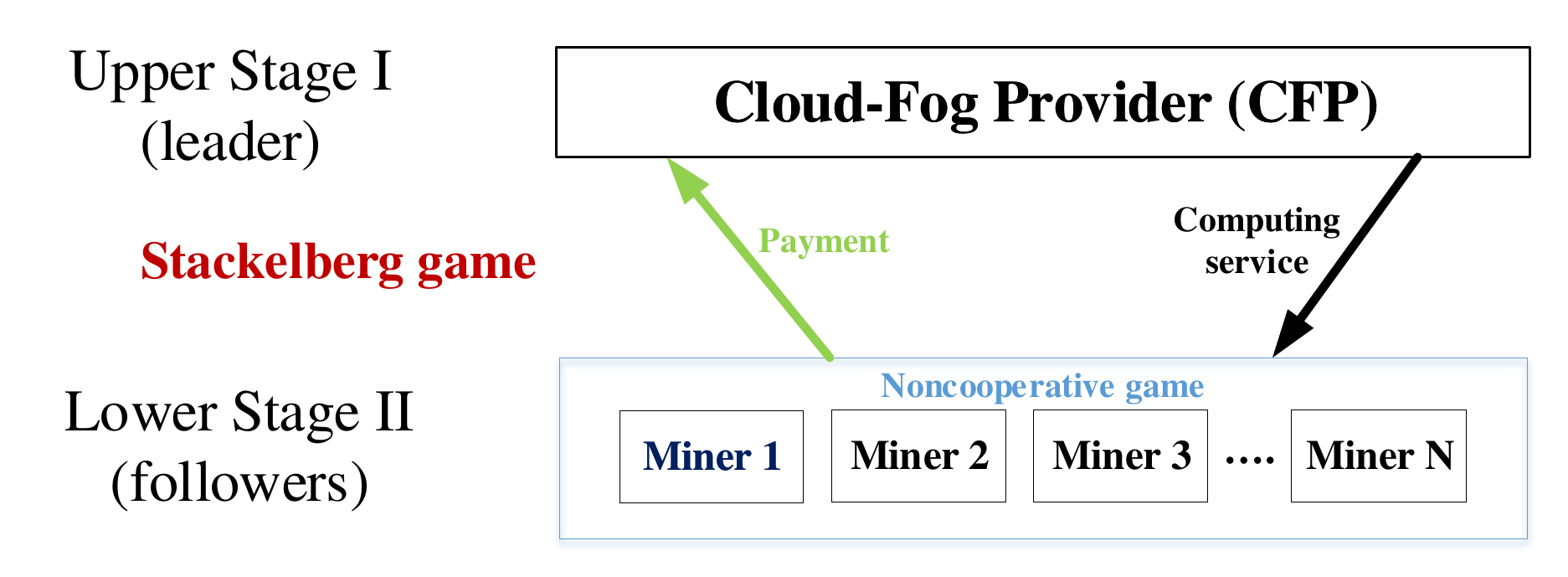}
\caption{\footnotesize{Two-stage Stackelberg game model of the interactions among the CFP and miners in the blockchain network.}}\label{Fig:GameModel}
\end{figure}
The interaction between the CFP and miners can be modeled as a two-stage Stackelberg game, as illustrated in Fig.~\ref{Fig:GameModel}. The CFP, i.e., the leader, sets the price in the upper Stage I. The miners, i.e., the followers, decide on their optimal computing service demand for offloading in the lower Stage II, being aware of the price set by the CFP. By using backward induction, we formulate the optimization problems for the leader and followers as follows.
\subsubsection{Miners' mining strategies in Stage II}
Given the pricing of the CFP and other miners' strategies, the miner $i$ determines its computing service demand for its hash power maximizing the expected utility which is given as:
\begin{equation}\label{Eq:Utility}
u_i({x_i},{{\bf{x}}_{ - i}},{p_i}) = (R + r{t_i})\frac{{{x_i}}}{{\sum\nolimits_{j\in \cal N} {{x_j}} }}{e^{ - \lambda z{t_i}}} - {p_i}{x_i},
\end{equation}
where $p_i$ is the price per unit for service demand of miner $i$. The miner sub-game problem can be written as follows:

{\bf{Problem 1. (Miner $i$ sub-game):}}
\begin{equation}
\begin{aligned}
& \underset{x_i}{\text{maximize}}
& & u_i({x_i},{{\bf{x}}_{ - i}},{p_i}) \\
& \text{subject to}
& & x_i \in [\underline x ,\overline x].\\
\end{aligned}
\end{equation}
\subsubsection{CFP's pricing strategies in Stage I}
The profit of the CFP is the revenue obtained from charging the miners for computing service minus the service cost. The service cost is directly related to the time that the miner takes to mine a block, the cost of electricity, $c$, and the other cost that is a function of the service demand $x_i$. Therefore, the CFP decides the pricing within the strategy space $\{{\bf p}= [p_i]_{i\in \cal N}:0\le p_i \le \overline p\}$ to maximize its profit which is represented as:
\begin{equation}\label{Eq:Profit}
\Pi ( {\mathbf{p}}, {\mathbf{x}} ) = \sum\nolimits_{i \in \cal N}{p_i}{x_i} - \sum\nolimits_{i \in \cal N}cT{x_i}.
\end{equation}
Note that practically the price is bounded by maximum price constraint that is denoted by
$\overline p$. Then, the profit maximization problem of the CFP is formulated as follows.

{\bf{Problem 2. (CFP sub-game):}}
\begin{equation}
\begin{aligned}
& \underset{\bf p}{\text{maximize}}
& & \Pi(\bf p, x) \\
& \text{subject to}
& & 0\le p_i \le \overline p.\\
\end{aligned}
\end{equation}

Problem 1 and Problem 2 together form the Stackelberg game, and the objective of this game is to find the Stackelberg equilibrium. The Stackelberg equilibrium ensures that the profit of the CFP is maximized given that the miners generate their demands following the best responses, i.e., the Nash equilibrium. This means that the demands from the miners maximize the utility. In our problem, the Stackelberg equilibrium can be written as follows.
\begin{definition}
Let $\bf x^*$ and $\bf p^*$ denote the optimal service demand vector of all the miners and optimal unit price vector of computing service, respectively. Then, the point $(\bf x^*, p^*)$ is the Stackelberg equilibrium if the following conditions,
\begin{equation}
\Pi ({\bf p^*},{\bf x^*}) \ge \Pi ({\bf p},{\bf x^*})
\end{equation}
and
\begin{equation}
u_i ( x^*_i, {\mathbf{x}}^*_{-i}, {\mathbf{p}}^* ) \ge {u_i}({x_i},{\bf x}_{ - i}^*,{\bf p^*}),\forall {x_i} \ge 0, \forall i
\end{equation}
are satisfied, where ${\bf x}_{-i}^*$ is the best response service demand vector for all the miners except miner $i$.
\end{definition}


Note that the same or different prices can be applied to the miners, which we refer to them as the uniform and discriminatory pricing schemes, respectively. In the following, we investigate these two pricing schemes for resource management in blockchain networks. The Stackelberg equilibrium ensures that the profit of the CFP is maximized given that the miners generate their demands following the best responses, i.e., the Nash equilibrium. This means that the demands from the miners maximize the utility. The Stackelberg equilibrium under the uniform pricing scheme contains only one single price that the CFP imposes to the miners identically. On the contrary, the equilibrium under the discriminatory pricing scheme contains different prices, each of which the CFP imposes to each miner separately.

The significance of each pricing scheme is as follows. Under the uniform pricing scheme, the equilibrium ensures a fair price applied to all miners. The miners are indifferent to choose the services. However, the CFP has limited degree of freedom to maximize its profit. By contrast, under the discriminatory pricing scheme, the CFP can customize the price for each miner, matching with the miner's demand and preference. As such, the profit obtained under the discriminatory pricing scheme is expected to be superior to that of the uniform pricing scheme in terms of the higher profit for the CFP.


\section{Equilibrium Analysis for Cloud/Fog Computing Resource Management}\label{Sec:Solution}
In this section, we propose the uniform pricing and discriminatory pricing schemes for resource management in blockchain application involving PoW assisted by the CFP. We then analyze the optimal service demand of miners as well as the profit maximization of the CFP under both pricing schemes.
\subsection{Uniform Pricing Scheme}\label{SubSec:Uniform}
We first consider the uniform pricing scheme, in which the CFP charges all the miners the same unit price for their computing service demand, i.e., $p_i = p, \forall i$. Given the payoff functions defined in Section~\ref{Sec:Model}, we use backward induction to analyze the Stackelberg game.
\subsubsection{Stage II: Miners' Demand Game}\label{SubsubSec:UtilityMaximization_Uniform}
Given the price $p$ decided by the CFP, in Stage II, the miners compete with each other to maximize their own utility by choosing their individual service demand, which forms the noncooperative Miners' Demand Game (MDG)~$\mathcal{G}^u= \{\mathcal{N},\{x_i\}_{i \in \mathcal{N}},\{u_i\}_{i \in \mathcal{N}}\}$, where $\cal N$ is the set of miners, $\{x_i\}_{i \in \mathcal{N}}$ is the strategy set, and $u_i$ is the utility, i.e., payoff, function of miner $i$. Specifically, each miner $i \in \cal N$ selects its strategy to maximize its utility function $u_i(x_i, {\bf{x}}_{-i}, p)$. We next analyze the existence and uniqueness of the Nash equilibrium in the MDG.

\begin{definition}
A demand vector ${\bf{x}}^* = (x^*_1, \ldots , x^*_N)$ is the Nash equilibrium of the MDG~$\mathcal{G}^u= \{\mathcal{N},\{x_i\}_{i \in \mathcal{N}},\{u_i\}_{i \in \mathcal{N}}\}$, if, for every miner $i\in \cal N$, $u_i(x^*_i, {\mathbf{x}}^*_{-i}, p) \ge u_i({x_i}', {\mathbf{x}}^*_{-i}, p)$ for all ${x_i}' \in [\underline x ,\overline x]$, where $u_i(x_i, {\bf x}_{-i})$ is the resulting utility of the miner $i$, given the other miners' demand ${\bf x}_{-i}$.
\end{definition}

\begin{figure*}[ht]\footnotesize
\begin{eqnarray}\label{Eq:BestResponse}
{x_i}^* = \mathscr F_i ({{\bf x}}) = \begin{cases}
\underline x, & \sqrt {\frac{{(R + r{t_i})\sum\limits_{i \ne j} {{x_j}} }}{{p{e^{ \lambda z{t_i}}}}}} - \sum\limits_{i \ne j} {{x_j}}< \underline x \cr \sqrt {\frac{{(R + r{t_i})\sum\limits_{i \ne j} {{x_j}} }}{{p{e^{  \lambda z{t_i}}}}}} - \sum\limits_{i \ne j} {{x_j}}, &\underline x \le \sqrt {\frac{{(R + r{t_i})\sum\limits_{i \ne j} {{x_j}} }}{{p{e^{ \lambda z{t_i}}}}}} - \sum\limits_{i \ne j} {{x_j}} \le \overline x \cr \overline x, &\sqrt {\frac{{(R + r{t_i})\sum\limits_{i \ne j} {{x_j}} }}{{p{e^{ \lambda z{t_i}}}}}} - \sum\limits_{i \ne j} {{x_j}}> \overline x\end{cases}.
\end{eqnarray}
\hrulefill
\end{figure*}

\begin{theorem}
A Nash equilibrium exists in MDG~$\mathcal{G}^u= \{\mathcal{N},\{x_i\}_{i \in \mathcal{N}},\{u_i\}_{i \in \mathcal{N}}\}$.
\end{theorem}
\begin{proof}
Firstly, the strategy space for each miner is defined to be $[\underline x ,\overline x]$, which is a non-empty, convex, compact subset of the Euclidean space. From~(\ref{Eq:Utility}), $u_i$ is apparently continuous in $[\underline x ,\overline x]$. Then, we take the first order and second order derivatives of~(\ref{Eq:Utility}) with respect to $x_i$ to prove its concavity, which can be written as follows:
\begin{equation}\label{Eq:FirstOrderUtility}
\frac{{\partial {u_i}}}{{\partial {x_i}}} = (R + r{t_i}){e^{ - \lambda z{t_i}}}\frac{{\partial {\alpha_i}}}{{\partial {x_i}}} - p,
\end{equation}
\begin{equation}
\frac{{{\partial ^2}{u_i}}}{{\partial {x_i}^2}} = (R + r{t_i}){e^{ - \lambda z{t_i}}}\frac{{{\partial ^2}{\alpha_i}}}{{\partial {x_i}^2}} < 0,
\end{equation}
where $\frac{{\partial {\alpha_i}}}{{\partial {x_i}}} = \frac{{\sum\nolimits_{i \ne j} {{x_j}} }}{{{{\left( {\sum\nolimits_{i\in \cal N} {{x_j}} } \right)}^2}}} > 0$, and $\frac{{{\partial ^2}{\alpha_i}}}{{\partial {x_i}^2}} = - 2\frac{{\sum\nolimits_{i \ne j} {{x_j}} }}{{{{\left( {\sum\nolimits_{i \in \cal N} {{x_j}} } \right)}^3}}} < 0$.

Therefore, we have proved that $u_i$ is strictly concave with respect to $x_i$. Accordingly, the Nash equilibrium exists in this noncooperative MDG $\mathcal{G}^u$~\cite{han2012game}. The proof is now completed.
\end{proof}
Further, based on the first order derivative condition, we have
\begin{equation}\label{Eq:FirstOrderDerivative}
\frac{{\partial {u_i}}}{{\partial {x_i}}} = (R + r{t_i}){e^{ - \lambda z{t_i}}}\frac{{\partial {\alpha_i}}}{{\partial {x_i}}} - p = 0,
\end{equation}
and we obtain the best response function of miner $i$ by solving~(\ref{Eq:FirstOrderDerivative}), as shown in~(\ref{Eq:BestResponse}).

\begin{theorem}
The uniqueness of the Nash equilibrium in the noncooperative MDG is guaranteed given the following condition
\begin{equation}\label{Eq:Condition}
\frac{{2(N - 1){e^{ \lambda z{t_i}}}}}{{R + r{t_i}}} < \sum\limits_{j \in \cal N} {\frac{{{e^{ \lambda z{t_j}}}}}{{R + r{t_j}}}}
\end{equation}
is satisfied.
\end{theorem}
\begin{proof}
Let $\bf x^*$ denote the Nash equilibrium of the MDG. By definition, the Nash equilibrium needs to satisfy ${\bf x} = \mathscr F(\bf x)$, in which $\mathscr F({\bf x}) = (\mathscr F_1({\bf x}), \mathscr F_2({\bf x}), \ldots, \mathscr F_N({\bf x}))$. In particular, $\mathscr F_i({\bf x})$ is the best response function of miner $i$, given the demand strategies of other miners. The uniqueness of the Nash equilibrium can be proved by showing that the best response function of miner $i$, i.e., as given in~(\ref{Eq:BestResponse}), is the standard function~\cite{han2012game}.
\begin{definition}
A function $\mathscr F(\bf x)$ is a standard function when the following properties are guaranteed~\cite{han2012game}:\\
$(1)$ Positivity: $\mathscr F(\bf x)>0$;\\
$(2)$ Monotonicity: If $\bf x \le x'$, then $\mathscr F(\bf x) \le \mathscr F(\bf x')$;\\
$(3)$ Scalability: For all $\phi >1$, $\phi \mathscr F(\bf x) > \mathscr F(\phi \bf x)$.\\
\end{definition}
Firstly, for the positivity, under the condition in~(\ref{Eq:Condition}), we have (from Lemma~1)
\begin{equation}
\sum\limits_{i \ne j} {{x_j}} < \frac{{R + r{t_i}}}{{4p{e^{ \lambda z{t_i}}}}} < \frac{{R + r{t_i}}}{{p{e^{ \lambda z{t_i}}}}},
\end{equation}
then we can conclude that
\begin{equation}
\sum\limits_{i \ne j} {{x_j}} < \sqrt {\frac{{(R + r{t_i})\sum\limits_{i \ne j} {{x_j}} }}{{p{e^{ \lambda z{t_i}}}}}}.
\end{equation}
Thus, we can prove that
\begin{equation}
{\mathscr F}_i({\bf x}) = \sqrt {\frac{{(R + r{t_i})\sum\limits_{i \ne j} {{x_j}} }}{{p{e^{ \lambda z{t_i}}}}}} - \sum\limits_{i \ne j} {{x_j}} > 0,
\end{equation}
which is the positivity condition.
\begin{figure*}[ht]\footnotesize
\begin{eqnarray}\label{Eq:Monotonicity}
{\cal F}_i({\bf{x}}') - {\cal F}_i({\bf{x}}) &=& \sqrt {\frac{{(R + r{t_i})\sum\limits_{i \ne j} {{x_j^\prime} } }}{{p{e^{ \lambda z{t_i}}}}}} - {\sum\limits_{i \ne j} {{x_j^\prime}} } - \sqrt {\frac{{(R + r{t_i})\sum\limits_{i \ne j} {{x_j}} }}{{p{e^{ \lambda z{t_i}}}}}} - \sum\limits_{i \ne j} {{x_j}} \nonumber \\
&=& \left( {\sqrt {\frac{{(R + r{t_i})}}{{p{e^{ \lambda z{t_i}}}}}} - \sqrt {{{\sum\limits_{i \ne j} {{x_j^\prime}} } }} - \sqrt {\sum\limits_{i \ne j} {{x_j}} } } \right) \left( {\sqrt {{{\sum\limits_{i \ne j} {{x_j^\prime}} } }} - \sqrt {\sum\limits_{i \ne j} {{x_j}} } } \right).
\end{eqnarray}
\begin{eqnarray}\label{Eq:Scalability}
\phi {\cal F}_i({\bf x}) - {\cal F}_i(\phi {\bf x}) &=& \phi \sqrt {\frac{{(R + r{t_i})\sum\limits_{i \ne j} {{x_j}} }}{{p{e^{  \lambda z{t_i}}}}}} -\phi \sum\limits_{i \ne j} {{x_j}} - \sqrt {\frac{{(R + r{t_i})\sum\limits_{i \ne j} {\phi {x_j}} }}{{p{e^{  \lambda z{t_i}}}}}} - \sum\limits_{i \ne j} {\phi {x_j}}\nonumber \\
 &=& \left( {\phi - \sqrt \phi } \right)\sqrt {\frac{{(R + r{t_i})\sum\limits_{i \ne j} {{x_j}} }}{{p{e^{  \lambda z{t_i}}}}}} > 0, \forall \phi>1.
\end{eqnarray}
\hrulefill
\end{figure*}
Secondly, we prove the monotonicity of~(\ref{Eq:BestResponse}). Let $\bf x' > x$, we can further simplify the expression of ${\cal F}_i({\bf{x}}') - {\cal F}_i({\bf{x}})$, which is shown in~(\ref{Eq:Monotonicity}). In particular, we have ${\sqrt {{{\sum\limits_{i \ne j} {{x_j^\prime}} } }} - \sqrt {\sum\limits_{i \ne j} {{x_j}} } }>0$, and we can easily verify that
\begin{multline}
{\sqrt {\frac{{R + r{t_i}}}{{p{e^{ \lambda z{t_i}}}}}} - \sqrt {{{\sum\limits_{i \ne j} {{x_j^\prime}} } }} - \sqrt {\sum\limits_{i \ne j} {{x_j}} } } \in \\ \left( {\sqrt {\frac{{R + r{t_i}}}{{p{e^{ \lambda z{t_i}}}}}} - 2\sqrt {{{\sum\limits_{i \ne j} {{x_j^\prime}} } }} ,\sqrt {\frac{{R + r{t_i}}}{{p{e^{ \lambda z{t_i}}}}}} - 2\sqrt {\sum\limits_{i \ne j} {{x_j}} } } \right).
\end{multline}
Under the condition in~(\ref{Eq:Condition1}), we can prove that
\begin{equation}
{\sqrt {\frac{{R + r{t_i}}}{{p{e^{ \lambda z{t_i}}}}}} - 2\sqrt {\sum\limits_{i \ne j} {{x_j}}}}>0, \forall x_j.
\end{equation}
Thus, the best response function of miner $i$ in~(\ref{Eq:BestResponse}) is always positive.

At last, as for scalability, we need to prove that $\phi \mathscr F(x) > \mathscr F(\phi x)$, for $\lambda >1$. The steps of proving the positivity of $\phi {\cal F}(x) - {\cal F}(\phi x)$ are shown in~(\ref{Eq:Scalability}). Therefore, $\phi \mathscr F(x) > \mathscr F(\phi x)$ is always satisfied for $\phi >1$. Until now, we have proved that the best response function in~(\ref{Eq:BestResponse}) satisfies three properties described in Definition~2. Therefore, the Nash equilibrium of MDG~$\mathcal{G}^u= \{\mathcal{N},\{x_i\}_{i \in \mathcal{N}},\{u_i\}_{i \in \mathcal{N}}\}$ is unique. The proof is now completed.
\end{proof}

\begin{theorem}
The unique Nash equilibrium for miner $i$ in the MDG is given by
\begin{equation}\label{Eq:MDGequilibruim}
{x_i}^* = \frac{{N - 1}}{{\sum\limits_{j\in \cal N} {\frac{{p{e^{ \lambda z{t_j}}}}}{{R + r{t_j}}}} }} - {\left( {\frac{{N - 1}}{{\sum\limits_{j\in \cal N} {\frac{{p{e^{ \lambda z{t_j}}}}}{{R + r{t_j}}}} }}} \right)^2}\frac{{p{e^{ \lambda z{t_i}}}}}{{R + r{t_i}}}, \forall i,
\end{equation}
provided that the condition in~(\ref{Eq:Condition}) holds.

\end{theorem}
\begin{proof}
According to~(\ref{Eq:FirstOrderUtility}), for each miner $i$, we have the mathematical expression
\begin{equation}
\frac{{\sum\limits_{i \ne j} {{x_j}} }}{{{{\left( {\sum\limits_{j\in \cal N}
 {{x_j}} } \right)}^2}}} = \frac{{p{e^{ \lambda z{t_i}}}}}{{R + r{t_i}}}.
\end{equation}
Then, we calculate the summation of this expression for all the miners as follows:
\begin{equation}
\frac{{(N - 1)\sum\limits_{j \in \cal N} {{x_j}} }}{{{{\left( {\sum\limits_{j\in \cal N} {{x_j}} } \right)}^2}}} = \sum\limits_{i\in \cal N} {\frac{{p{e^{ \lambda z{t_i}}}}}{{R + r{t_i}}}},
\end{equation}
which means $\frac{{(N - 1)}}{{\sum\limits_{j\in \cal N} {{x_j}} }} = \sum\limits_{i \in \cal N} {\frac{{p{e^{ \lambda z{t_i}}}}}{{R + r{t_i}}}}$. Thus, we have
\begin{equation}\label{Eq:Summative}
\sum\limits_{j \in \cal N} {{x_j}} = \frac{{N - 1}}{{\sum\limits_{i \in \cal N} {\frac{{p{e^{ \lambda z{t_i}}}}}{{R + r{t_i}}}} }}.
\end{equation}
Recall from~(\ref{Eq:BestResponse}), according to the first order derivative condition, we have
\begin{equation}\label{Eq:FirstDerivative}
\sum\limits_{j \in \cal N} {{x_j}} = \sqrt {\frac{{(R + r{t_i})\sum\limits_{i \ne j} {{x_j}} }}{{p{e^{ \lambda z{t_i}}}}}}.
\end{equation}
By substituting~(\ref{Eq:FirstDerivative}) into~(\ref{Eq:Summative}), we have
\begin{equation}
\frac{{N - 1}}{{\sum\limits_{i \in \cal N} {\frac{{p{e^{ \lambda z{t_i}}}}}{{R + r{t_i}}}} }} = \sqrt {\frac{{R + r{t_i}}}{{p{e^{ \lambda z{t_i}}}}}\left( {\frac{{N - 1}}{{\sum\limits_{i\in \cal N} {\frac{{p{e^{ \lambda z{t_i}}}}}{{R + r{t_i}}}} }} - {x_i}} \right)}.
\end{equation}
After squaring both sides, we have ${\left( {\frac{{N - 1}}{{\sum\limits_{i\in \cal N} {\frac{{p{e^{ \lambda z{t_i}}}}}{{R + r{t_i}}}} }}} \right)^2} = \frac{{R + r{t_i}}}{{p{e^{ \lambda z{t_i}}}}}\left( {\frac{{N - 1}}{{\sum\limits_{i\in \cal N} {\frac{{p{e^{ \lambda z{t_i}}}}}{{R + r{t_i}}}} }} - {x_i}} \right)$. With simple transformations, we obtain the Nash equilibrium for miner $i$ as shown in~(\ref{Eq:MDGequilibruim}).
\end{proof}

\begin{lemma}
Given
\begin{equation}\label{Eq:Condition1}
\frac{{2(N - 1){e^{ \lambda z{t_i}}}}}{{R + r{t_i}}} < \sum\limits_{i\in \cal N} {\frac{{{e^{ \lambda z{t_i}}}}}{{R + r{t_i}}}},
\end{equation}
the following condition
\begin{equation}\label{Eq:lemmacondition1}
\sum\limits_{i \ne j} {{x_j}} < \frac{{R + r{t_i}}}{{4p{e^{  \lambda z{t_i}}}}}
\end{equation}
is satisfied.
\end{lemma}
\begin{proof}
According to~(\ref{Eq:MDGequilibruim}) and~(\ref{Eq:Summative}), we can obtain
\begin{equation}\label{Eq:lemmacondition2}
\sum\limits_{j \ne i} {{x_j}} = {\left( {\frac{{N - 1}}{{\sum\limits_{j \in \cal N} {\frac{{p{e^{ \lambda z{t_j}}}}}{{R + r{t_j}}}} }}} \right)^2}\frac{{p{e^{ \lambda z{t_i}}}}}{{R + r{t_i}}}.
\end{equation}
After substituting~(\ref{Eq:lemmacondition1}) into~(\ref{Eq:lemmacondition2}), we have
\begin{equation}\label{Eq:lemmacondition3}
\frac{{2(N - 1)p{e^{ \lambda z{t_i}}}}}{{R + r{t_i}}} < \sum\limits_{i \in \cal N} {\frac{{p{e^{ \lambda z{t_i}}}}}{{R + r{t_i}}}},
\end{equation}
which means that the condition in~(\ref{Eq:Condition1}) needs to be ensured. On the contrary, if the condition in~(\ref{Eq:Condition1}) holds, then, the condition in~(\ref{Eq:lemmacondition3}) is satisfied. The proof is now completed.
\end{proof}

Generally, we can use the best-response dynamics for obtaining the Nash equilibrium of the N-player noncooperative game in Stage II~\cite{han2012game}. In the following, we analyze the profit maximization of the CFP in Stage I under uniform pricing.

\subsubsection{Stage I: CFP's Profit Maximization}\label{SubsubSec:ProfitMaximization_Uniform}

Based on the Nash equilibrium of the computing service demand in the MDG~$\mathcal{G}^u= \{\mathcal{N},\{x_i\}_{i \in \mathcal{N}},\{u_i\}_{i \in \mathcal{N}}\}$ in Stage II, the leader of the Stackelberg game, i.e., the CFP, can optimize its pricing strategy in Stage I to maximize its profit defined in~(\ref{Eq:Profit}). Thus, the optimal pricing can be formulated as an optimization problem. By substituting~(\ref{Eq:MDGequilibruim}) into~(\ref{Eq:Profit}), the profit maximization of the CFP is simplified as follows:
\begin{equation}\label{Eq:Profit2}
\begin{aligned}
& \underset{p>0}{\text{maximize}}
& & \Pi(p) = (p - cT)\frac{{N - 1}}{{\sum\limits_{j \in \cal N} {\frac{{p{e^{ \lambda z{t_j}}}}}{{R + r{t_j}}}} }} \\
& \text{subject to}
& & 0\le p \le \overline p.\\
\end{aligned}
\end{equation}
\begin{theorem}
Under uniform pricing, the CFP achieves the globally optimal profit, i.e., profit maximization, under the unique optimal price.
\end{theorem}
\begin{proof}
From~(\ref{Eq:Profit2}), we have
\begin{equation}
\Pi(p) = \frac{{p - cT}}{p}\frac{{N - 1}}{{\sum\limits_{j \in \cal N} {\frac{{{e^{ \lambda z{t_j}}}}}{{R + r{t_j}}}} }}.
\end{equation}
The first and second derivatives of profit $\Pi(p)$ with respect to price $p$ are given as follows:
\begin{equation}\label{Eq:FirstOrderProfit}
\frac{{d \Pi(p) }}{{d p}} = \frac{{cT}}{{{p^2}}}\frac{{N - 1}}{{\sum\limits_{j \in \cal N} {\frac{{{e^{  \lambda z{t_j}}}}}{{R + r{t_j}}}} }}
\end{equation}
and
\begin{equation}\label{Eq:SecondOrderProfit}
\frac{{{d ^2}\Pi(p) }}{{d {p^2}}} = - \frac{{2cT}}{{{p^2}}}\frac{{N - 1}}{{\sum\limits_{j \in \cal N} {\frac{{{e^{ \lambda z{t_j}}}}}{{R + r{t_j}}}} }}< 0.
\end{equation}
Due to the negativity of~(\ref{Eq:SecondOrderProfit}), the strict concavity of the objective function is ensured. Thus, the CFP is able to achieve the maximum profit with the unique optimal price. The proof is now completed.
\end{proof}
Note that the profit maximization defined in~(\ref{Eq:Profit2}) is a convex optimization problem, and thus it can be solved by standard convex optimization algorithms, e.g., gradient assisted binary search. Under uniform pricing, we have proved that the Nash equilibrium in Stage II is unique and the optimal price in Stage I is also unique. Thus, we can conclude that the Stackelberg equilibrium is unique and accordingly the best-response dynamics algorithm can achieve this unique Stackelberg equilibrium~\cite{han2012game}.

\subsection{Discriminatory Pricing Scheme}\label{SubSec:DiscriminatoryPricing}
Then, we consider the discriminatory pricing scheme, in which the CFP is able to set different unit prices of service demand for different miners. Again, we use the backward induction to analyze the optimal service demand of miners and the profit maximization of the CFP.

\subsubsection{Stage II: Miners' Demand Game}\label{SubsubSec:UtilityMaximization_Discriminatory}
Under discriminatory pricing scheme, the strategy space of the CFP becomes $\{{\bf p}= [p_i]_{i\in \cal N}:0\le p_i \le \overline p\}$. Recall that we prove the existence and uniqueness of MDG~$\mathcal{G}^u= \{\mathcal{N},\{x_i\}_{i \in \mathcal{N}},\{u_i\}_{i \in \mathcal{N}}\}$, given the fixed price from the CFP. Thus, under discriminatory pricing, the existence and uniqueness of the MDG can be still guaranteed. With minor change from Theorem~3, we have the following theorem immediately.
\begin{theorem}
Under uniform pricing, the unique Nash equilibrium demand of miner $i$ can be obtained as follows:
\begin{equation}\label{Eq:MDGequilibruim2}
{x_i}^* = \frac{{N - 1}}{{\sum\limits_{j\in \cal N}{\frac{{p_j{e^{  \lambda z{t_j}}}}}{{R + r{t_j}}}} }} - {\left( {\frac{{N - 1}}{{\sum\limits_{j \in \cal N} {\frac{{p_j{e^{  \lambda z{t_j}}}}}{{R + r{t_j}}}} }}} \right)^2}\frac{{p_i{e^{  \lambda z{t_i}}}}}{{R + r{t_i}}}, \forall i,
\end{equation}
if the following condition
\begin{equation}\label{Eq:Condition3}
\frac{{2(N - 1){p_i}{e^{  \lambda z{t_i}}}}}{{R + r{t_i}}} < \sum\limits_{j\in \cal N} {\frac{{{p_j}{e^{  \lambda z{t_j}}}}}{{R + r{t_j}}}}
\end{equation}
holds.
\end{theorem}
\begin{proof}
The steps of proof are similar to those in the case of uniform pricing as shown in Section~\ref{SubsubSec:UtilityMaximization_Uniform}, and thus we omit them for brevity.
\end{proof}

We next analyze the profit maximization of the CFP in Stage I under discriminatory pricing to further investigate the Stackelberg equilibrium.

\subsubsection{Stage I: CFP's Profit Maximization}\label{SubsubSec:ProfitMaximization_Discriminatory}

Similar to that in Section~\ref{SubsubSec:ProfitMaximization_Uniform}, we analyze the profit maximization with the analytical result from Theorem~5, i.e., the Nash equilibrium of the computing service demand in Stage II. After substituting~(\ref{Eq:MDGequilibruim2}) into~(\ref{Eq:Profit}), we have the following optimization,
\begin{equation}\label{Eq:Profit3}
\begin{aligned}
& \underset{\bf p>0}{\text{maximize}}
& & \Pi({\bf p}) = \sum\limits_{i \in \cal N}\left(p_i - cT\frac{{N - 1}}{{\sum\limits_{j\in \cal N} {\frac{{p_j{e^{ \lambda z{t_j}}}}}{{R + r{t_j}}}} }}\right) \\
& \text{subject to}
& & 0\le p_i \le \overline p, \forall i.\\
\end{aligned}
\end{equation}

\begin{theorem}
$\Pi({\bf p})$ is concave on each $p_i$, when $\sum\limits_{i \ne j} {({a_i} + {a_j})\left( {1 - \frac{{N\frac{{{p_j}}}{{{a_j}}}}}{{\sum\limits_{j \in \cal N} {\frac{{{p_j}}}{{{a_j}}}} }}} \right)} \le 0$, and decreasing on each $p_i$ when $\sum\limits_{i \ne j} {({a_i} + {a_j})\left( {1 - \frac{{N\frac{{{p_j}}}{{{a_j}}}}}{{\sum\limits_{j \in \cal N} {\frac{{{p_j}}}{{{a_j}}}} }}} \right)} >0$, provided that the following condition
\begin{equation}\label{Eq:Condition4}
\frac{{{p_i}}}{{{a_i}}} \ge \frac{{\sum\limits_{j \in \cal N} {\frac{{{p_j}}}{{{a_j}}}} }}{{{{(N - 1)}^2}}}
\end{equation}
is satisfied, where $a_i =(R+rt_i)e^{-\lambda z t_i}$.
\end{theorem}

\begin{figure*}[ht]\footnotesize
\begin{eqnarray}\label{Eq:PriceFinal}
g({\bf p})
&=& \sum\limits_{j \ne h} {\left( {{a_h}\left( {1 - \frac{{{p_h}}}{{{a_h}}}\frac{{N - 1}}{{\sum\limits_{h \in \cal N} {\frac{{{p_h}}}{{{a_h}}}} }}} \right)\left( {1 - \frac{{{p_j}}}{{{a_j}}}\frac{{N - 1}}{{\sum\limits_{h \in \cal N} {\frac{{{p_h}}}{{{a_h}}}} }}} \right)} \right)}.
\end{eqnarray}
\begin{eqnarray}\label{Eq:PriceFinalFirstOrder}
\frac{{\partial g({\bf{p}})}}{{\partial {p_i}}} = \sum\limits_{j \ne i} {\left( {\left( {{a_i} + {a_j}} \right)\left( {\frac{{ - \frac{{N - 1}}{{{a_i}}}\sum\limits_{h \ne i} {\frac{{{p_h}}}{{{a_h}}}} }}{{{{\left( {\sum\limits_{h \in \cal N} {\frac{{{p_h}}}{{{a_h}}}} } \right)}^2}}}\left( {1 - \frac{{N - 1}}{{\sum\limits_{h \in \cal N} {\frac{{{p_h}}}{{{a_h}}}} }}\frac{{{p_j}}}{{{a_j}}}} \right) + \frac{{\frac{{N - 1}}{{{a_i}}}\frac{{{p_j}}}{{{a_j}}}}}{{{{\left( {\sum\limits_{h \in \cal N} {\frac{{{p_h}}}{{{a_h}}}} } \right)}^2}}}\left( {1 - \frac{{N - 1}}{{\sum\limits_{h \in \cal N} {\frac{{{p_h}}}{{{a_h}}}} }}\frac{{{p_i}}}{{{a_i}}}} \right)} \right)} \right)}.
\end{eqnarray}
\hrulefill
\end{figure*}
\begin{figure*}[ht]\footnotesize
\begin{eqnarray}\label{Eq:VIconvexset}
&&\sum\limits_{i \ne j} {\left( {({a_i} + {a_j})\left( {\sum\limits_{h \in \cal N} {\frac{{\lambda p_h^{'} + (1 - \lambda )p_h^{''}}}{{{a_h}}}} - N\frac{{\lambda p_j^{'} + (1 - \lambda )p_j^{''}}}{{{a_j}}}} \right)} \right)} \nonumber \\
&=& \sum\limits_{i \ne j} {\left( {({a_i} + {a_j})\left( {\lambda \sum\limits_{h \in \cal N} {\frac{{p_h^{'}}}{{{a_h}}}} - (1 - \lambda )\sum\limits_{h \in \cal N} {\frac{{p_h^{''}}}{{{a_h}}}} - \lambda N\frac{{p_j^{'}}}{{{a_j}}} - (1 - \lambda )N\frac{{p_j^{''}}}{{{a_j}}}} \right)} \right)} \nonumber \\
&=& \lambda \sum\limits_{i \ne j} {\left( {({a_i} + {a_j})\left( {\sum\limits_{h \in \cal N} {\frac{{p_h^{'}}}{{{a_h}}}} - N\frac{{p_j^{''}}}{{{a_j}}}} \right)} \right)} + (1 - \lambda )\sum\limits_{i \ne j} {\left( {({a_i} + {a_j})\left( {(1 - \lambda )\sum\limits_{h \in \cal N} {\frac{{p_h^{'}}}{{{a_h}}}} - N\frac{{p_j^{''}}}{{{a_j}}}} \right)} \right)} \le 0.
\end{eqnarray}
\hrulefill
\end{figure*}

\begin{figure*}[ht]\scriptsize
\begin{eqnarray}\label{Eq:Decreasing}
&&\frac{{\partial \Pi({\bf p}) }}{{{\partial p_i}}}= \sum\limits_{j \ne i} {\left( {\left( {{a_i} + {a_j}} \right)\left( {\frac{{\frac{{N - 1}}{{{a_i}}}\sum\limits_{h \ne i} {\frac{{{p_h}}}{{{a_h}}}} }}{{{{\left( {\sum\limits_{h \in \cal N} {\frac{{{p_h}}}{{{a_h}}}} } \right)}^2}}}\left( {1 - \frac{{N - 1}}{{\sum\limits_{h \in \cal N} {\frac{{{p_h}}}{{{a_h}}}} }}\frac{{{p_j}}}{{{a_j}}}} \right) + \frac{{\frac{{N - 1}}{{{a_i}{a_j}}}{p_j}}}{{{{\left( {\sum\limits_{h \in \cal N} {\frac{{{p_h}}}{{{a_h}}}} } \right)}^2}}}\left( {1 - \frac{{N - 1}}{{\sum\limits_{h \in \cal N} {\frac{{{p_h}}}{{{a_h}}}} }}\frac{{{p_i}}}{{{a_i}}}} \right)} \right)} \right)} + \frac{{\frac{{N - 1}}{{{a_i}}}cT}}{{{{\left( {\sum\limits_{h \in \cal N} {\frac{{{p_h}}}{{{a_h}}}} } \right)}^2}}} \nonumber \\
&\le& \frac{{\frac{{N - 1}}{{{a_i}}}}}{{{{\left( {\sum\limits_{h \in \cal N} {\frac{{{p_h}}}{{{a_h}}}} } \right)}^2}}}\left( {\sum\limits_{j \ne i} {\left( {\left( {{a_i} + {a_j}} \right)\left( { - \sum\limits_{h \in \cal N} {\frac{{{p_h}}}{{{a_h}}}} \left( {1 - \frac{{N - 1}}{{\sum\limits_{h \in \cal N} {\frac{{{p_h}}}{{{a_h}}}} }}\frac{{{p_j}}}{{{a_j}}}} \right) + \frac{{{p_j}}}{{{a_j}}}\left( {1 - \frac{{N - 1}}{{\sum\limits_{h \in \cal N} {\frac{{{p_h}}}{{{a_h}}}} }}\frac{{{p_i}}}{{{a_i}}}} \right)} \right)} \right)} + cT} \right) \nonumber \\
&=& \underbrace { - \frac{{\frac{{N - 1}}{{{a_i}}}}}{{{{\left( {\sum\limits_{h \in \cal N} {\frac{{{p_h}}}{{{a_h}}}} } \right)}^2}}}\sum\limits_{j \ne i} {\left( {\left( {{a_i} + {a_j}} \right)\left( { \sum\limits_{h \in \cal N} {\frac{{{p_h}}}{{{a_h}}}} \left( {1 - \frac{N}{{\sum\limits_{h \in \cal N} {\frac{{{p_h}}}{{{a_h}}}} }}\frac{{{p_j}}}{{{a_j}}}} \right)} \right)} \right)} }_{ < 0} + \frac{{\frac{{N - 1}}{{{a_i}}}}}{{{{\left( {\sum\limits_{h \in \cal N} {\frac{{{p_h}}}{{{a_h}}}} } \right)}^2}}}\left( {cT - \sum\limits_{j \ne i} {\left( {\left( {{a_i} + {a_j}} \right)\frac{{N - 1}}{{\sum\limits_{h \in \cal N} {\frac{{{p_h}}}{{{a_h}}}} }}\frac{{{p_i}}}{{{a_i}}}\frac{{{p_j}}}{{{a_j}}}} \right)} } \right) \nonumber \\
&=& - \frac{{\frac{{N - 1}}{{{a_i}}}}}{{{{\left( {\sum\limits_{h \in \cal N} {\frac{{{p_h}}}{{{a_h}}}} } \right)}^2}}}\sum\limits_{j \ne i} {\left( {\left( {{a_i} + {a_j}} \right)\left( { - \sum\limits_{h \in \cal N} {\frac{{{p_h}}}{{{a_h}}}} \left( {1 - \frac{N}{{\sum\limits_{h \in \cal N} {\frac{{{p_h}}}{{{a_h}}}} }}\frac{{{p_j}}}{{{a_j}}}} \right)} \right)} \right)} + \frac{{\frac{{N - 1}}{{{a_i}}}}}{{{{\left( {\sum\limits_{h \in \cal N} {\frac{{{p_h}}}{{{a_h}}}} } \right)}^2}}}\left( {cT - \sum\limits_{j \ne i} {\left( {\underbrace { {\frac{{{a_i} + {a_j}}}{{{a_j}}}}}_{ < 1}\frac{{N - 1}}{{\sum\limits_{h \in \cal N} {\frac{{{p_h}}}{{{a_h}}}} }}\frac{{{p_i}{p_j}}}{{{a_i}}}} \right)} } \right) \nonumber \\
&\le& - \frac{{\frac{{N - 1}}{{{a_i}}}}}{{{{\left( {\sum\limits_{h \in \cal N} {\frac{{{p_h}}}{{{a_h}}}} } \right)}^2}}}\sum\limits_{j \ne i} {\left( {\left( {{a_i} + {a_j}} \right)\left( { - \sum\limits_{h \in \cal N} {\frac{{{p_h}}}{{{a_h}}}} \left( {1 - \frac{N}{{\sum\limits_{h \in \cal N} {\frac{{{p_h}}}{{{a_h}}}} }}\frac{{{p_j}}}{{{a_j}}}} \right)} \right)} \right)} + \frac{{\frac{{N - 1}}{{{a_i}}}}}{{{{\left( {\sum\limits_{h \in \cal N} {\frac{{{p_h}}}{{{a_h}}}} } \right)}^2}}}\left( {cT - p_{\min} \frac{{N - 1}}{{\sum\limits_{h \in \cal N} {\frac{{{p_h}}}{{{a_h}}}} }}\frac{{N - 1{p_i}}}{{{a_i}}}} \right) \nonumber \\
&=& \underbrace { - \frac{{\frac{{N - 1}}{{{a_i}}}}}{{{{\left( {\sum\limits_{h \in \cal N} {\frac{{{p_h}}}{{{a_h}}}} } \right)}^2}}}\sum\limits_{j \ne i} {\left( {\left( {{a_i} + {a_j}} \right)\left( { - \sum\limits_{h \in \cal N} {\frac{{{p_h}}}{{{a_h}}}} \left( {1 - \frac{N}{{\sum\limits_{h \in \cal N} {\frac{{{p_h}}}{{{a_h}}}} }}\frac{{{p_j}}}{{{a_j}}}} \right)} \right)} \right)} }_{ < 0} + \underbrace {\frac{{\frac{{N - 1}}{{{a_i}}}}}{{{{\left( {\sum\limits_{h \in \cal N} {\frac{{{p_h}}}{{{a_h}}}} } \right)}^2}}}\left( {cT - {p_{\min }}\frac{{{{\left( {N - 1} \right)}^2}}}{{\sum\limits_{h \in \cal N} {\frac{{{p_h}}}{{{a_h}}}} }}\frac{{{p_i}}}{{{a_i}}}} \right)}_{ < 0} < 0.
\end{eqnarray}
\hrulefill
\end{figure*}
\begin{proof}
We firstly decompose the objective function in~(\ref{Eq:Profit3}) into two parts, namely, $\sum\limits_{i} c T{x_i^*}$ and ${\sum\limits_{i} p_i}{x_i^*}$. Then, we analyze the properties of each part. We define
\begin{equation}\label{Eq:cost}
f({\bf p}) = - cT{x_i^*} =  - cT\frac{{N - 1}}{{\sum\limits_{j \in \cal N} {\frac{{p_j{e^{ \lambda z{t_j}}}}}{{R + r{t_j}}}} }}.
\end{equation}
Let $a_j =(R+rt_j)e^{-\lambda z t_j}$, and we have $f({\bf p}) = \frac{{ - cT(N - 1)}}{{\sum\limits_{j \in \cal N} {\frac{{{p_j}}}{{{a_j}}}} }}$. Then, we obtain the first and the second partial derivatives of~(\ref{Eq:cost}) with respect to $p_i$ as follows.
\begin{equation}
\frac{{\partial f({\bf{p}})}}{{\partial {p_i}}} = \frac{{(N - 1)cT}}{{{a_i}{{\left( {\sum\limits_{j \in \cal N} {\frac{{{p_j}}}{{{a_j}}}} } \right)}^2}}},
\end{equation}
\begin{equation}
\frac{{{\partial ^2}f({\bf{p}})}}{{\partial {p_i}^2}} = \frac{{ - 2(N - 1)cT}}{{{a_i}^2{{\left( {\sum\limits_{j \in \cal N} {\frac{{{p_j}}}{{{a_j}}}} } \right)}^3}}}.
\end{equation}
Further, we have
\begin{equation}
\frac{{\partial f({\bf{p}})}}{{\partial {p_i}{p_j}}} = \frac{{ - 2(N - 1)cT}}{{{a_i}{a_j}{{\left( {\sum\limits_{j\in \cal N} {\frac{{{p_j}}}{{{a_j}}}} } \right)}^3}}}.
\end{equation}
Thus, we can obtain the Hessian matrix of $f({\bf p})$, which is expressed as:
\begin{equation}
{\nabla ^2}f({\bf p}) = \frac{{ - 2(N - 1)cT}}{{{{\left( {\sum\limits_{j\in \cal N} {\frac{{{p_j}}}{{{a_j}}}} } \right)}^3}}}\left[ {\begin{array}{*{20}{c}}
{\frac{1}{{{a_1}^2}}}&{\frac{1}{{{a_1}{a_2}}}}& \cdots &{\frac{1}{{{a_1}{a_N}}}}\\
{\frac{1}{{{a_2}{a_1}}}}&{\frac{1}{{{a_2}^2}}}& \cdots &{\frac{1}{{{a_2}{a_N}}}}\\
 \vdots & \vdots & \ddots & \vdots \\
{\frac{1}{{{a_N}{a_1}}}}&{\frac{1}{{{a_N}{a_2}}}}& \cdots &{\frac{1}{{{a_N}^2}}}
\end{array}} \right].
\end{equation}
For each $i \in \cal N$, we have ${\frac{1}{{{a_i}^2}}}>0$. Thus, the diagonal elements of the Hessian matrix are all larger than zero, and the principle minors are equal to zero. Therefore, the Hessian matrix of $f({\bf p})$ is semi-negative definite.

Then, we analyze the properties of ${\sum\limits_{i} p_i}{x_i^*}$. We first define
\begin{equation}\label{Eq:Price}
g({\bf p}) = \sum\limits_{i \in \cal N} {{p_i}} {x_i}^* = \frac{{\sum\limits_{j \ne i} {{a_i}{x_i}{x_j}} }}{{{{\left( {\sum\limits_{j \ne i} {{x_j}} } \right)}^2}}}.
\end{equation}
By substituting~(\ref{Eq:MDGequilibruim2}) into~(\ref{Eq:Price}), we can obtain the final expression for $g({\bf p})$, which can be rewritten as in~(\ref{Eq:PriceFinal}). Then, we derive the first order and the second partial derivatives of~(\ref{Eq:PriceFinal}) with respect to $p_i$ as shown in~(\ref{Eq:PriceFinalFirstOrder}) and~(\ref{Eq:PriceFinalSecondOrder}). Since we have $x_i = \frac{{N - 1}}{{\sum\limits_{h \in \cal N} {\frac{{{p_h}}}{{{a_h}}}} }} - \frac{{{p_i}}}{{{a_i}}}{\left( {\frac{{N - 1}}{{\sum\limits_{h \in \cal N} {\frac{{{p_h}}}{{{a_h}}}} }}} \right)^2} = \frac{{N - 1}}{{\sum\limits_{h \in \cal N} {\frac{{{p_h}}}{{{a_h}}}} }}\left( {1 - \frac{{N - 1}}{{\sum\limits_{h \in \cal N} {\frac{{{p_h}}}{{{a_h}}}} }}\frac{{{p_i}}}{{{a_i}}}} \right) >0 $, ${1 - \frac{{N - 1}}{{\sum\limits_{h \in \cal N} {\frac{{{p_h}}}{{{a_h}}}} }}\frac{{{p_i}}}{{{a_i}}}} >0$. When $\sum\limits_{i \ne j} {({a_i} + {a_j})\left( {1 - \frac{{N\frac{{{p_j}}}{{{a_j}}}}}{{\sum\limits_{j\in \cal N} {\frac{{{p_j}}}{{{a_j}}}} }}} \right)} \le 0$, it is observed that $\frac{{{\partial ^2}g({\bf{p}})}}{{\partial {p_i}^2}} < 0$, i.e., $g({\bf p})$ is concave on each $p_i$. Now we prove that $\Pi({\bf p})$ is a monotonically decreasing function with respect to $p_i$, when $\sum\limits_{i \ne j} {({a_i} + {a_j})\left( {1 - \frac{{N\frac{{{p_j}}}{{{a_j}}}}}{{\sum\limits_{j\in \cal N} {\frac{{{p_j}}}{{{a_j}}}} }}} \right)} > 0$. The steps are shown in~(\ref{Eq:Decreasing}), where ${{p_{\min }} = \min \{ {p_1},{p_2}, \ldots ,{p_N}\} }$. Practically, $p_{\min } > cT$. Thus, with some manipulations, we can prove $\frac{{\partial \Pi }}{{\partial {p_i}}}<0$ when $\sum\limits_{i \ne j} {({a_i} + {a_j})\left( {1 - \frac{{N\frac{{{p_j}}}{{{a_j}}}}}{{\sum\limits_{j \in \cal N} {\frac{{{p_j}}}{{{a_j}}}} }}} \right)} > 0$, if the condition in~(\ref{Eq:Condition4}) holds. The proof is now completed.
\end{proof}

\begin{theorem}
Under discriminatory pricing, the CFP achieves the profit maximization by finding the unique optimal pricing vector.
\end{theorem}
\begin{proof}
\begin{figure*}[ht]\footnotesize
\begin{eqnarray}\label{Eq:PriceFinalSecondOrder}
\frac{{{\partial ^2}g({\bf{p}})}}{{\partial {p_i}^2}}
&=& \sum\limits_{j \ne i} {\left( {\left( {{a_i} + {a_j}} \right)\left( {\frac{{2\frac{{N - 1}}{{{a_i}^2}}\sum\limits_{h \ne i} {\frac{{{p_h}}}{{{a_h}}}} }}{{{{\left( {\sum\limits_{h \in \cal N} {\frac{{{p_h}}}{{{a_h}}}} } \right)}^3}}}\left( {1 - 2\frac{{N - 1}}{{\sum\limits_{h \in \cal N} {\frac{{{p_h}}}{{{a_h}}}} }}\frac{{{p_j}}}{{{a_j}}}} \right) - \frac{{2\frac{{N - 1}}{{{a_i}^2}}\frac{{{p_j}}}{{{a_j}}}}}{{{{\left( {\sum\limits_{h \in \cal N} {\frac{{{p_h}}}{{{a_h}}}} } \right)}^3}}}\left( {1 - \frac{{N - 1}}{{\sum\limits_{h \in \cal N} {\frac{{{p_h}}}{{{a_h}}}} }}\frac{{{p_i}}}{{{a_i}}}} \right)} \right)} \right)}.
\end{eqnarray}
\begin{eqnarray}
\frac{{{\partial ^2}g({\bf{p}})}}{{\partial {p_i}^2}} &=& \frac{{2\frac{{N - 1}}{{{a_i}^2}}\sum\limits_{h \ne i} {\frac{{{p_h}}}{{{a_h}}}} }}{{{{\left( {\sum\limits_{h \in \cal N} {\frac{{{p_h}}}{{{a_h}}}} } \right)}^3}}}\sum\limits_{j \ne i} {\left( {\left( {{a_i} + {a_j}} \right)\left( {1 - 2\frac{{N - 1}}{{\sum\limits_{h \in \cal N} {\frac{{{p_h}}}{{{a_h}}}} }}\frac{{{p_j}}}{{{a_j}}}} \right)} \right)} - \frac{{2\frac{{N - 1}}{{{a_i}^2}}}}{{{{\left( {\sum\limits_{h \in \cal N} {\frac{{{p_h}}}{{{a_h}}}} } \right)}^3}}}\sum\limits_{j \ne i} {\left( {\left( {{a_i} + {a_j}} \right)\frac{{{p_j}}}{{{a_j}}}\left( {1 - \frac{{N - 1}}{{\sum\limits_{h \in \cal N} {\frac{{{p_h}}}{{{a_h}}}} }}\frac{{{p_i}}}{{{a_i}}}} \right)} \right)}\nonumber \\
&\le& \frac{{2\frac{{N - 1}}{{{a_i}^2}}\sum\limits_{h \ne i} {\frac{{{p_h}}}{{{a_h}}}} }}{{{{\left( {\sum\limits_{h \in \cal N} {\frac{{{p_h}}}{{{a_h}}}} } \right)}^3}}}\underbrace {\mathop {\sum\limits_{j \ne i} {\left( {\left( {{a_i} + {a_j}} \right)\left( {1 - \frac{N}{{\sum\limits_{h \in \cal N} {\frac{{{p_h}}}{{{a_h}}}} }}\frac{{{p_j}}}{{{a_j}}}} \right)} \right)} }\limits}_{ \le 0} - \frac{{2\frac{{N - 1}}{{{a_i}^2}}}}{{{{\left( {\sum\limits_{h \in \cal N} {\frac{{{p_h}}}{{{a_h}}}} } \right)}^3}}}\sum\limits_{j \ne i} {\left( {\left( {{a_i} + {a_j}} \right)\frac{{{p_j}}}{{{a_j}}}\underbrace {\left( {1 - \frac{{N - 1}}{{\sum\limits_{h \in \cal N} {\frac{{{p_h}}}{{{a_h}}}} }}\frac{{{p_i}}}{{{a_i}}}} \right)}_{ \ge 0}} \right)}.
\end{eqnarray}
\hrulefill
\end{figure*}
From Theorem~6, we know that $\Pi({\bf p})$ is concave on each $p_i$, when $\sum\limits_{i \ne j} {({a_i} + {a_j})\left( {1 - \frac{{N\frac{{{p_j}}}{{{a_j}}}}}{{\sum\limits_{j\in \cal N} {\frac{{{p_j}}}{{{a_j}}}} }}} \right)} \le 0$, and decreasing on each $p_i$ when $\sum\limits_{i \ne j} {({a_i} + {a_j})\left( {1 - \frac{{N\frac{{{p_j}}}{{{a_j}}}}}{{\sum\limits_{j\in \cal N} {\frac{{{p_j}}}{{{a_j}}}} }}} \right)} > 0$. In other words, when $\Pi({\bf p})$ is concave on $p_i$, $p_i$ needs to be smaller than a certain threshold, and $\Pi({\bf p})$ is decreasing on $p_i$ when $p_i$ is larger than this threshold. Then, it can be concluded that if the price is higher than the threshold, the miner is not willing to purchase the computing service from the CFP. Therefore, we know that the optimal value of profit of the CFP, i.e., $\Pi^*({\bf p})$ is achieved in the concave parts when $\sum\limits_{i \ne j} {({a_i} + {a_j})\left( {1 - \frac{{N\frac{{{p_j}}}{{{a_j}}}}}{{\sum\limits_{j \in \cal N} {\frac{{{p_j}}}{{{a_j}}}} }}} \right)} \le 0$. Clearly, the maximization of profit $\Pi({\bf p})$ is achieved either in the boundary of domain area or in the local maximization point. Since we know that the optimal value of profit, i.e., $\Pi^*({\bf p})$ is achieved in the interior area, and thus $\bf p^*$ exists. In the following, we prove that there exists at most one optimal solution by using Variational Inequality theory~\cite{scutari2010convex}, from which the uniqueness of the optimal solution, i.e., the Stackelberg equilibrium, follows.

Let the set ${\cal K}=\bigg \{ {\bf p} = [p_1, \ldots, p_N]^{\top}\bigg| {\sum\limits_{i \ne j} {({a_i} + {a_j})\left( {1 - \frac{{N\frac{{{p_j}}}{{{a_j}}}}}{{\sum\limits_{j \in \cal N} {\frac{{{p_j}}}{{{a_j}}}} }}} \right)} \le 0}, \forall i \in \cal N \bigg\}$. The constraint can be rewritten as follows:
\begin{equation}
\sum\limits_{i \ne j} {\left( {({a_i} + {a_j})\left( {\sum\limits_{h \in \cal N} {\frac{{{p_h}}}{{{a_h}}}} - N\frac{{{p_j}}}{{{a_j}}}} \right)} \right)} \le 0.
\end{equation}
Thus, we redefine the set $\cal K$ as $ \bigg \{ {\bf p}=[p_1, \ldots, p_N]^{\top}$ $\bigg| {\sum\limits_{i \ne j} {\left( {({a_i} + {a_j})\left( {\sum\limits_{h \in \cal N} {\frac{{{p_h}}}{{{a_h}}}} - N\frac{{{p_j}}}{{{a_j}}}} \right)} \right)} \le 0}, \forall i \in \cal N \bigg\}$. Then, we formulate an equivalent problem to~(\ref{Eq:Profit3}) as follows:
\begin{equation}\label{Eq:optimizationtoVI}
\begin{aligned}
& \underset{\bf p>0}{\text{minimize}}
& & -\Pi({\bf p}) \\
& \text{subject to}
& & {\bf p} \in {\cal K}.\\
\end{aligned}
\end{equation}
Let $F({\bf p})=\nabla \left( { - \Pi ({\bf{p}})} \right) = - {\left[ {{\nabla _{{p_i}}}\Pi } \right]^{\top}_{i \in \cal N}}$. Accordingly, the optimization problem in~(\ref{Eq:optimizationtoVI}) is equivalent to find a point set ${\bf p^*}\in \cal K$, such that $({\bf p} - {\bf p^*})F({\bf p^*})\ge 0 ,\forall {\bf p} \in \cal K$, which is the Variational Inequality (VI) problem: VI$({\cal K}, F)$.

\begin{definition}
If $F$ is strictly monotone on $\cal K$, then VI$({\cal K}, F)$ has at most one solution, where ${\cal K}\in \mathbb{R}^N$ is a convex closed set, and the mapping $F: {\cal K}\mapsto \mathbb{R}^N$ is continuous~\cite{scutari2010convex}.
\end{definition}

Let $\lambda \in (0,1)$, ${\bf p'}, {\bf p''} \in \cal K$, it can be concluded that $\lambda {\bf{p'}} + (1 - \lambda ){\bf{p'}}\in \cal K$, which is shown in~(\ref{Eq:VIconvexset}). Accordingly, $\cal K$ is a convex and closed set. To prove that the mapping $F: {\cal K}\mapsto \mathbb{R}^N$ is strictly monotone on $\cal K$, we check the positivity of $({\bf p'} - {\bf p''})^{\top}(F({\bf p'})-F({\bf p''})), \forall {\bf p'}, {\bf p''} \in \cal K$ and ${\bf p'} \ne {\bf p''}$. We know
\begin{multline}
({\bf{p'}} - {\bf{p''}})^\top (F({\bf{p'}}) - F({\bf{p''}})) = \\ \sum\limits_{i \in \cal N} {\left( { (p'_i - p''_i)\left({  -\left. {{\nabla _{{p_i}}}\Pi} \right|_{p_i=p'_i} + \left. {{\nabla _{{p_i}}}\Pi} \right|_{p_i=p''_i} }\right) } \right)},
\end{multline}
and from Theorem~6, we have
\begin{equation}
\frac{{{\partial ^2}\Pi ({\bf{p}})}}{{\partial {p_i}^2}} = \frac{{{\partial ^2}(f({\bf{p}}) + g({\bf{p}}))}}{{\partial {p_i}^2}} < 0.
\end{equation}
Thus, ${{\nabla _{{p_i}}}\Pi }$ is decreasing on each $p_i$, and $-{{\nabla _{{p_i}}}\Pi }$ is increasing on each $p_i$. It can be concluded that
\begin{equation}
- {\left. {{\nabla _{{p_i}}}\Pi } \right|_{{p_i} = {p'_i}}} + {\left. {{\nabla _{{p_i}}}\Pi } \right|_{{p_i} = {p''_i}}} = \left\{ {\begin{array}{*{20}{c}}
{ \ge 0,{p'_i} \ge {p''_i}}\\
{ < 0,{p'_i} < {p''_i}}	.
\end{array}} \right.
\end{equation}
Then, we have
\begin{equation}\label{Eq:VIconstraint}
{\left( { (p'_i - p''_i)\left({  -\left. {{\nabla _{{p_i}}}\Pi} \right|_{p_i=p'_i} + \left. {{\nabla _{{p_i}}}\Pi} \right|_{p_i=p''_i} }\right) } \right)} \ge 0, \forall i \in \cal N,
\end{equation}
and we know ${\bf p'} \ne {\bf p''}$, and accordingly there exists at least one $j \in \cal N$ which satisfies the constraint in~(\ref{Eq:VIconstraint}). Therefore, we have proved that $F$ is strictly monotone on $\cal K$ and continuous. Until now, we have proved that VI$({\cal K}, F)$ has at most one solution according to Definition~4 in~\cite{scutari2010convex}. Thus, the equivalent problem admits at most one optimal solution. Since we know the existence of a single optimal solution, and thus the uniqueness of the optimal solution is validated. The proof is now completed.
\end{proof}

Similar to that in Section~\ref{SubSec:Uniform}, we can apply the low-complexity gradient based searching algorithm to achieve the maximized profit $\Pi({\bf p})$ of the CFP. In particular, we adopt Algorithm~1 to obtain the unique Stackelberg equilibrium, under which the CFP achieves the profit maximization according to Theorem~7. The basic description is explained as follows: for the given prices imposed by the CFP, the followers' sub-game is solved first. After substituting the best responses of the followers' sub-game into the leader sub-game, the optimal prices can be obtained by a gradient-based algorithm. The similar algorithm can be used for uniform pricing as well.
\begin{algorithm}\scriptsize
 \caption{Gradient iterative algorithm to find Stackelberg equilibrium under discriminatory pricing}
 \begin{algorithmic}[1]
 \STATE \textbf{Initialization:} \\
 Select initial input ${\bf p}= [p_i]_{i\in \cal N}$ where $p_i \in [0, \overline p]$, $k\leftarrow 1$, precision threshold $\varepsilon$;
 \REPEAT
  \STATE Each miner $i$ decides its computing service demand $x_i^{[k]}$ based on~(\ref{Eq:BestResponse});
  \STATE CFP updates the prices using a gradient assisted searching algorithm, i.e.,
  \begin{equation}
  {\bf p}(t + 1) = {\bf p}(t) + \mu \nabla \Pi ({\bf p}(t)),
  \end{equation}
  where $\mu$ is the step size of the price update and $ \mu \nabla \Pi ({\bf p}(t))$ is the gradient with $\frac{{\partial \Pi ({\bf{p}}(t))}}{{\partial {\bf{p}}(t)}}$. The price information is sent to all miners;
 \STATE $k\leftarrow k + 1 $;
 \UNTIL{$\frac{{{{\left\| {{{{{\bf p}}}^{[k]}} - {{{{\bf p}}}^{[k- 1]}}} \right\|}_1}}}{{{{\left\| {{{{{\bf p}}}^{[k-1]}}} \right\|}_1}}} < \varepsilon$}\\
 \STATE \textbf{Output:} optimal demand ${\bf{x^*}}^{[k]}$ and optimal price ${\bf p}^{*{[k]}}$.
 \end{algorithmic}\label{algorithm}
\end{algorithm}
\begin{figure}[t]
\centering
\includegraphics[width=.4\textwidth]{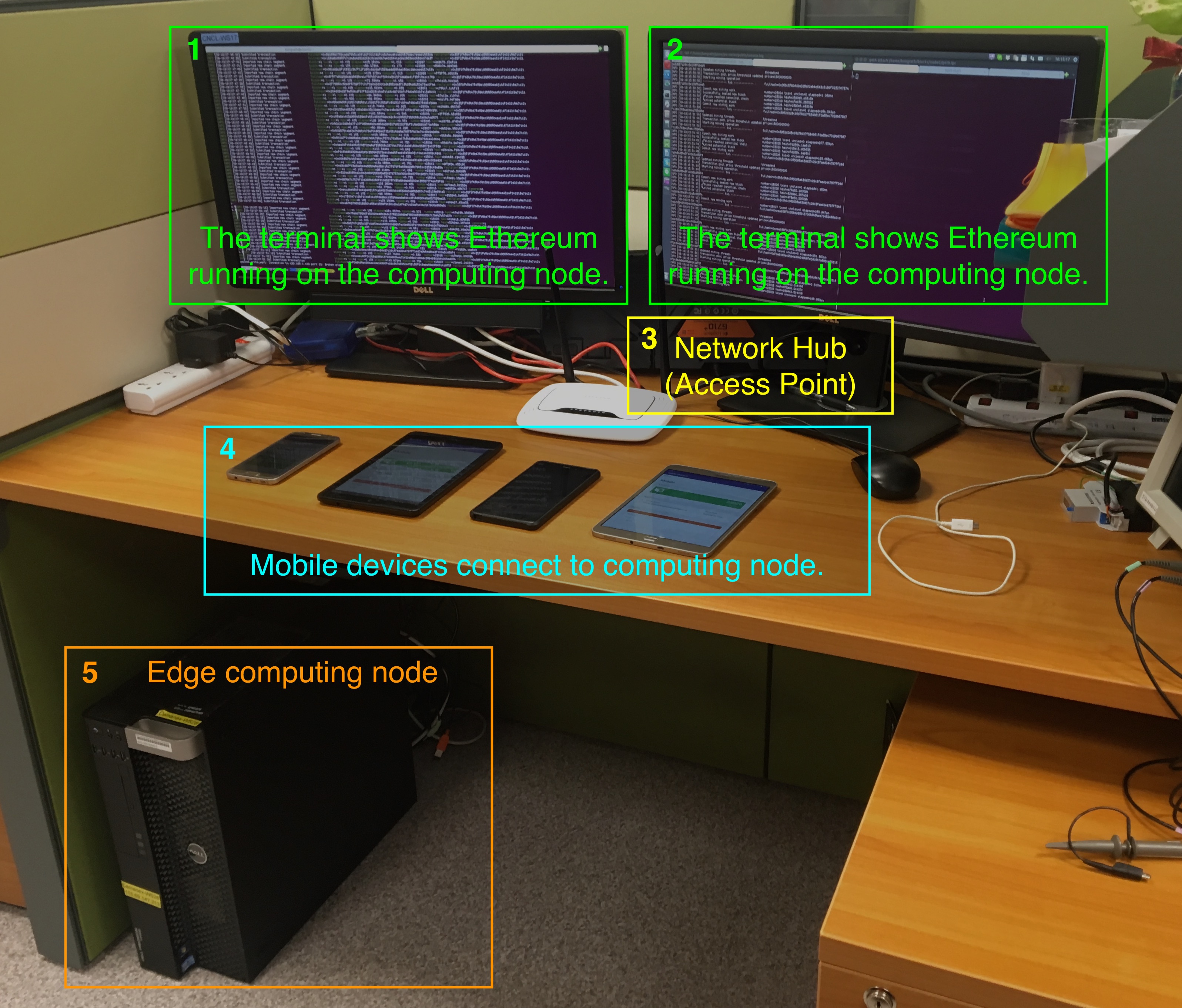}
\caption{\footnotesize{Real mobile blockchain mining experimental setup with Ethereum which is a popular open ledger.}}\label{Fig:Experiment}
\end{figure}

\begin{figure}[t]
\centering
\includegraphics[width=.4\textwidth]{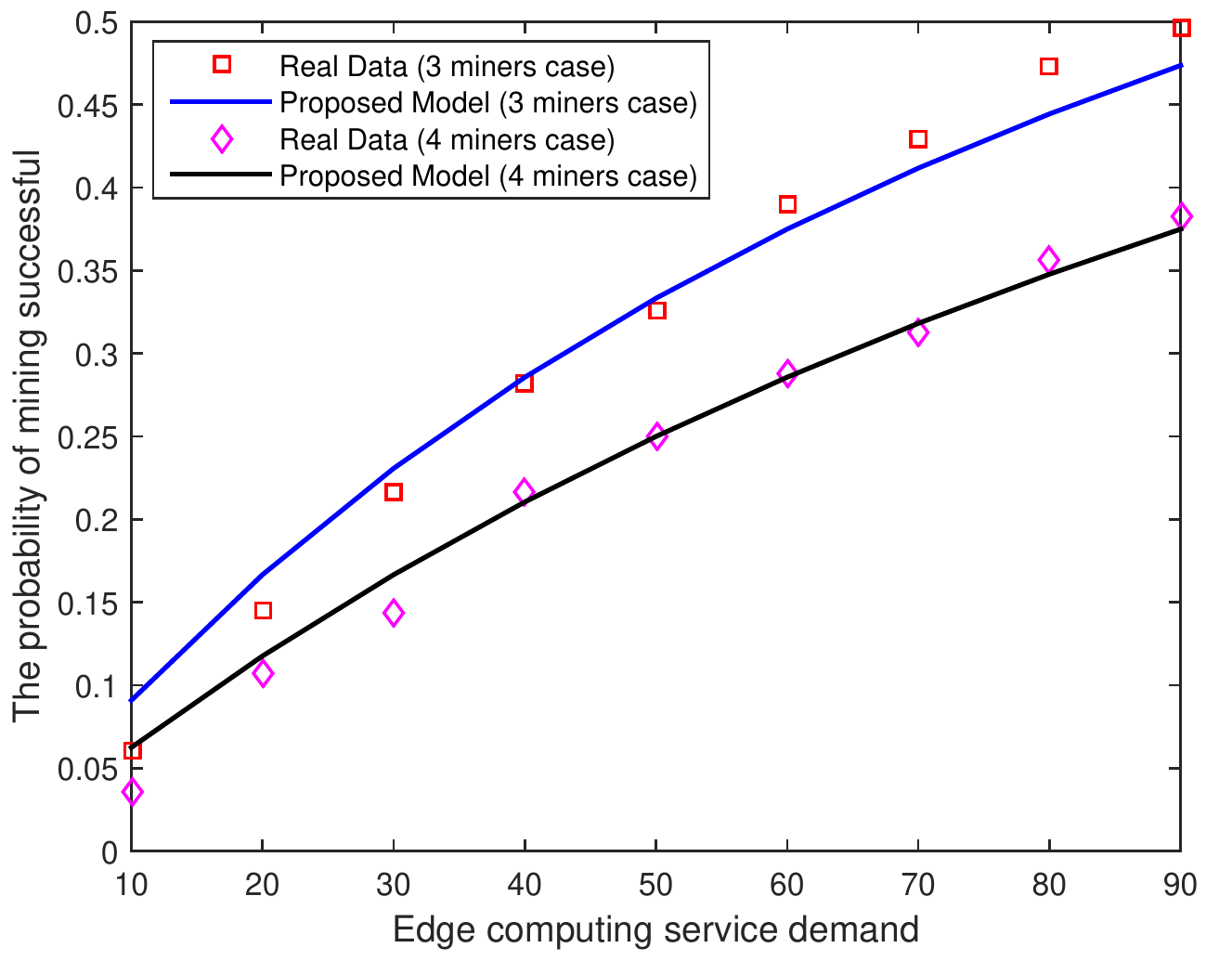}
\caption{\footnotesize{The comparison of real experiment results with our proposed model.}}\label{Fig:Comparison}
\end{figure}

\section{Performance Evaluation}\label{Sec:Simulation}
In this section, we first perform the real experiment on the PoW-based blockchain mining to validate the proposed utility function of the miner. Then, we conduct the extensive numerical simulations to evaluate the performance of our proposed price-based computing resource management to support blockchain application involving PoW.

\subsection{Environmental Setup}\label{Subsec:RealEnvironment}

We first set up the real blockchain mining experiment based on Ethereum and consider the smart phones as limited devices, as illustrated in Fig.~\ref{Fig:Experiment}. The experiment is performed on a workstation with Intel Xeon CPU E5-1630, and android devices (smart phones) installing a mobile blockchain client application. The mobile blockchain client application is implemented by the Android Studio and Software Development Kits (SDK) tools. All transactions are created by the mobile blockchain client application\footnote{In our experiment, each mobile device sends transactions to the server, and the size of each transaction is around 1 kilobyte~\cite{suankaewmanee2018performance}. Then, the server will collect and pack all the transactions into a block and proceed to solve the proof-of-work puzzle, where each block consists of block information and hash numbers. As mentioned in the paper, the number of transactions in each mined block is 10, and thus the size of the data from mobile device sent to the server is approximately 10 kilobytes in total. Likewise, the size of a block including 10 hash numbers that is sent from the server to the mobile device is around 1.5 kilobytes. The detailed description can be found in our previous work~\cite{suankaewmanee2018performance}.}. Each miner's working environment has one CPU core as its processor. The miner's processor and its CPU utilization rate are generated and managed by the Docker platform~\cite{Docker}. The mobile device of each miner has installed Ubuntu 16.04 LTS (Xenial Xerus) and Go-Ethereum~\cite{Ethereum} as the operation system and the blockchain framework, respectively.

In Fig.~\ref{Fig:Experiment}, from Box $1$ and $2$, the screen of computer terminal shows that the Ethereum is running on the host, i.e., edge device (Box $5$). The mobile devices in Box $4$ are connected to the edge computing node through network hub (Box $3$) using mobile blockchain client application. The basic steps can be implemented as follows. The mobile users, i.e., miners use the Android device to connect to the edge computing node through network hub, i.e., access point. Then, the miners can request the service from edge node, and mine the block with the assistance of Ethereum service provided accordingly.

We create $1000$ blocks employing Node.js and use the mobile device to mine these blocks in the experiment. We consider two cases with three miners and four miners. In the three-miner case, we first fix the other two miners' service demand (CPU utilization) at $40$ and $60$, and then vary one miner's service demand. In the four-miner case, we first fix other three miners' service demand as $40$, $50$ and $60$, and then vary one miner's service demand. For our experiment, the number of transactions in each mined block is $10$, i.e., the size of block is the same. The comparison of the real experimental results and our proposed analytical model is shown in Fig.~\ref{Fig:Comparison}. As expected, there is not much difference between the real results and our analytical model. This is because the probability that the miner successfully mines the block is directly proportional to its relative computing power when the block size are identical. Note that the delay effects are negligible. In the sequel, we present the numerical results to evaluate the performance of the proposed price-based computing resource management for supporting blockchain application involving PoW.

\subsection{Numerical Results}\label{Subsec:Numerical}

To illustrate the impacts of different parameters from the proposed model on the performance, we consider a group of $N$ miners, e.g., mobile users in the blockchain application involving PoW assisted by the CFP. We assume the size of a block mined by miner $i$ follows the normal distribution ${\cal N}(\mu_t, \sigma^2)$. The default parameter values are set as follows: $\underline x = 10^{-2}$, $\overline x = 100$, $\overline p =100$, $\mu_t = 200$, $\sigma^2 = 5$, $R = 10^4$, $r = 20$, $z = 5 \times 10^{-3}$, $c=10^{-3}$ and $N=100$. Further, we employ the `fix' function in MATLAB to round each $t_i$ to the nearest integer toward zero. Note that some of these parameters are varied according to the evaluation scenarios. We evaluate the performance of uniform pricing and discriminatory pricing in the following.

\subsubsection{Investigation on total service demand of miners and the profit of the CFP}\label{Subsubsec:demandandprofit}
\paragraph{The comparison of uniform pricing and discriminatory pricing}\label{Subsubsubsec:comparison}
\begin{figure}[t]
\centering
\includegraphics[width=.4\textwidth]{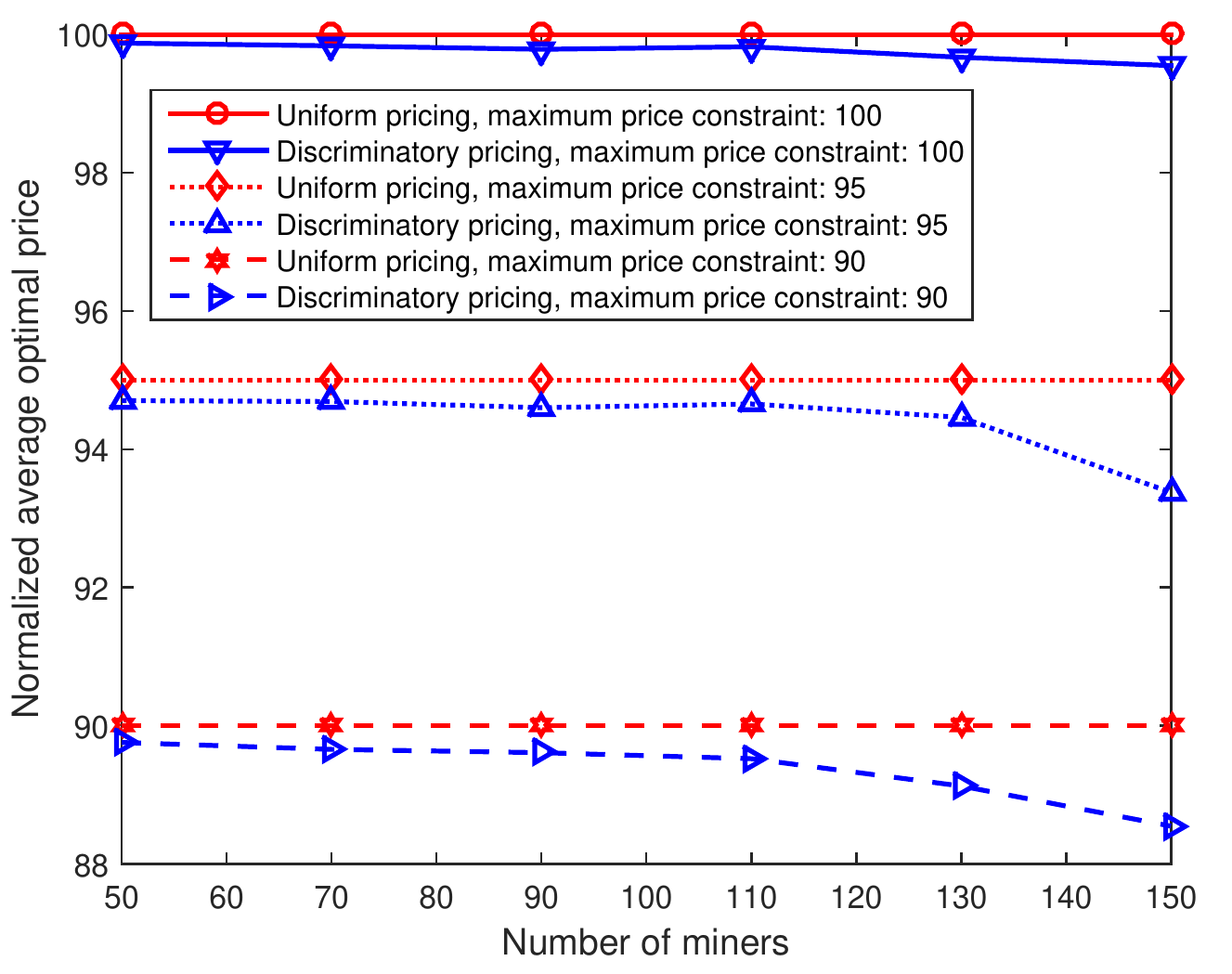}
\caption{\footnotesize{Normalized average optimal price versus the number of miners.}}\label{Fig:price_comparison}
\end{figure}

We first address the comparison of uniform pricing and discriminatory pricing schemes. Figure.~\ref{Fig:price_comparison} demonstrates the comparison of the normalized average optimal price under two proposed pricing schemes. It is worth noting that the optimal price under uniform pricing is the same as the maximum price, which can be explained by~(\ref{Eq:FirstOrderProfit}). Specifically, the expression in~(\ref{Eq:FirstOrderProfit}) is always positive, and thus the profit of the CFP increases with the increase of price. This means that the maximum price is the optimal value for profit maximization of the CFP under uniform pricing. Thus, we have the following conclusion: the CFP intends to set the maximum possible value as optimal price under uniform pricing. This conclusion is still useful even when the CFP does not have the complete information about the miners.

Further, we find that the average optimal price of discriminatory pricing is slightly lower than that of uniform pricing. The intuition is that, under under discriminatory pricing, the CFP can set different unit prices of service demand for different miners. For the details of operation of discriminatory pricing, we conduct the case study in Section~\ref{Subsubsec:price}. In this case, the CFP can significantly encourage the higher total service demand from miners and achieve greater profit gain under discriminatory pricing, which is also consistent with the following results. As shown in Figs.~\ref{Fig:utility_profit_number}-\ref{Fig:utility_profit_delay}, in all cases, the total service demand from miners and the profit of the CFP under the uniform pricing scheme is slightly smaller than that under the discriminatory pricing scheme.

From Fig.~\ref{Fig:utility_profit_number}, we find that when $\sigma^2$ decreases, the results under uniform pricing scheme is close to that under discriminatory pricing. This is because the heterogeneity of miners in blockchain is reduced as $\sigma^2$ decreases. We may consider one symmetric case, where the miners are homogeneous with the same size of blocks to mine, i.e., $\sigma^2 = 0$. In this case, the discriminatory pricing scheme yields the same results as those of the uniform pricing scheme.
%
%
\begin{figure*}
\begin{minipage}[t]{0.328\textwidth}
\centering
\includegraphics[width=2.2in]{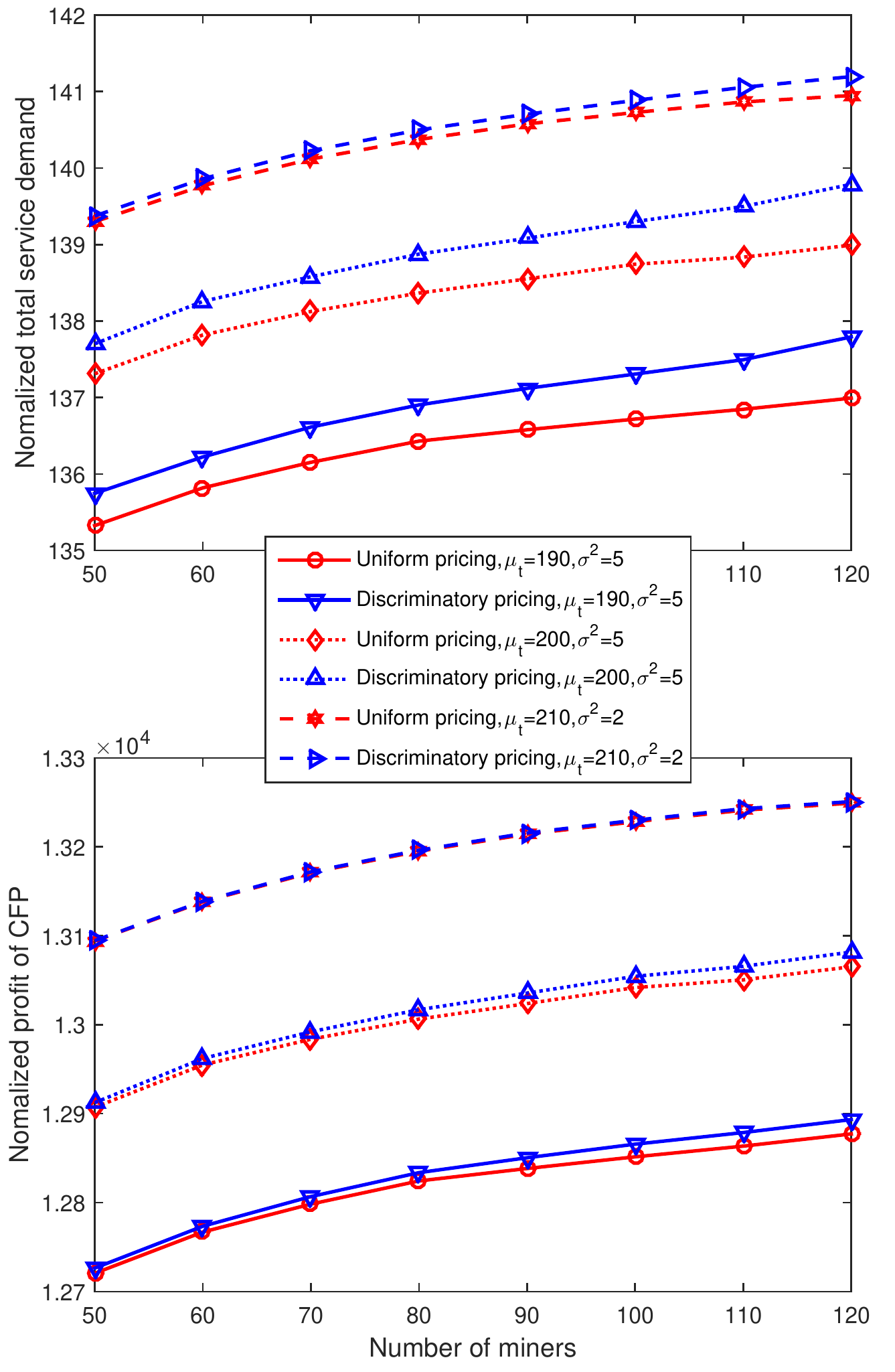}
\caption{\footnotesize{Normalized total service demand of miners and the profit of the CFP versus the number of miners.}}\label{Fig:utility_profit_number}
\end{minipage}
\begin{minipage}[t]{0.328\textwidth}
\centering
\includegraphics[width=2.2in]{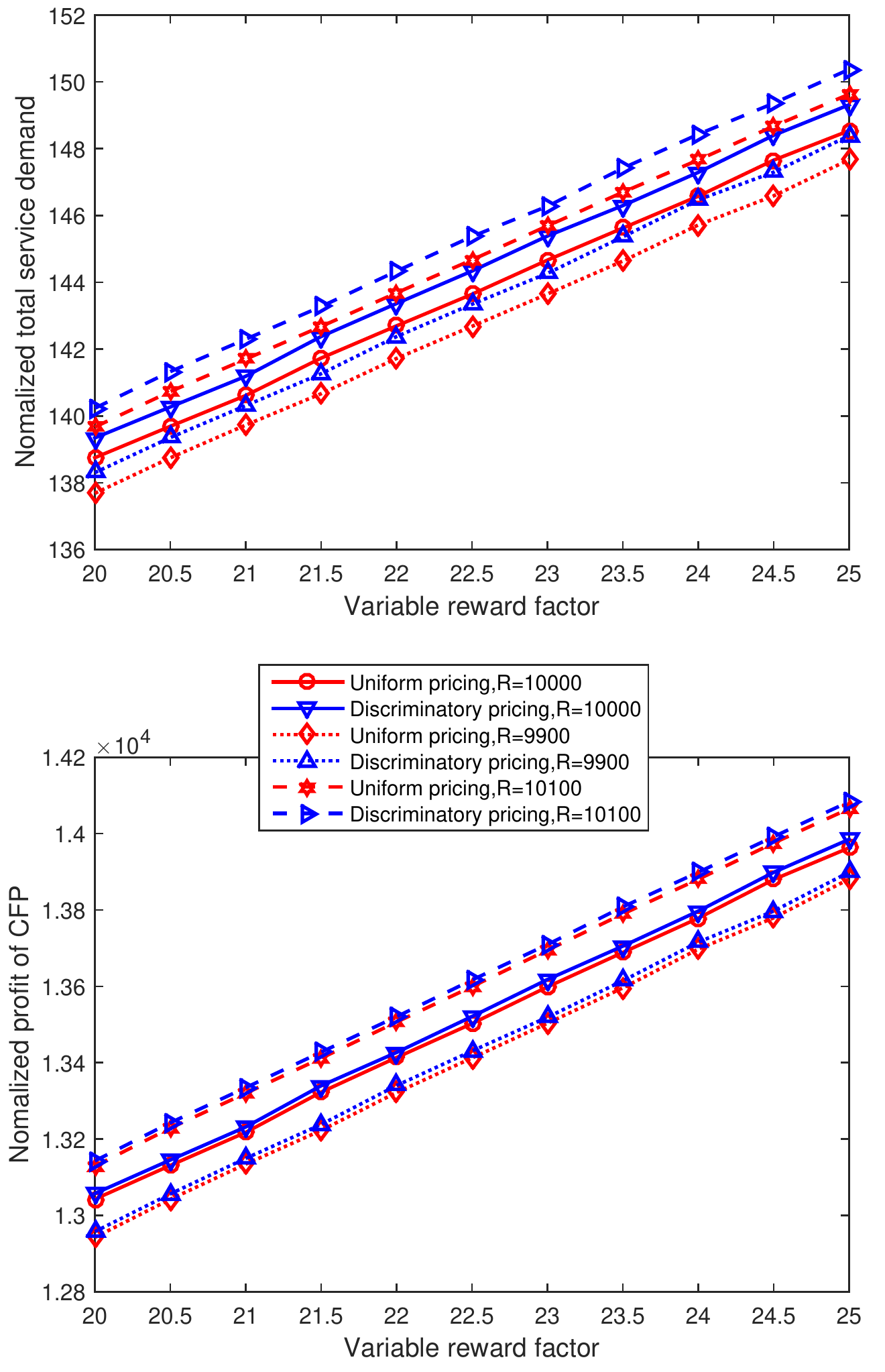}
\caption{\footnotesize{Normalized total service demand of miners and the profit of the CFP versus the variable reward factor.}}\label{Fig:utility_profit_reward}
\end{minipage}
\begin{minipage}[t]{0.328\textwidth}
\centering
\includegraphics[width=2.2in]{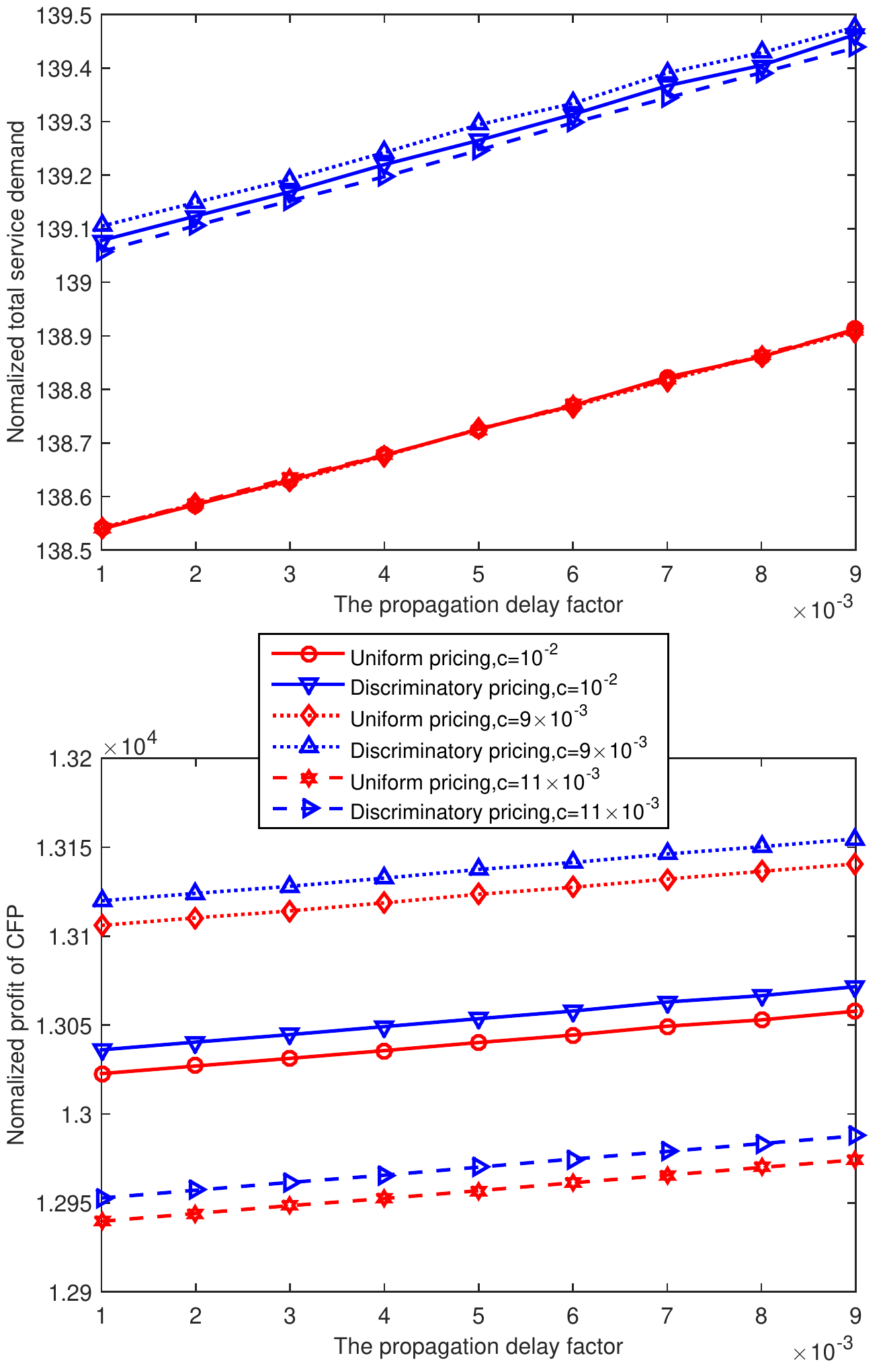}
\caption{\footnotesize{Normalized total service demand of miners and the profit of the CFP versus the propagation delay factor.}}\label{Fig:utility_profit_delay}
\end{minipage}
\end{figure*}
\paragraph{The impacts of the number of miners}\label{Subsubsubsec:number}

We next evaluate the impacts brought by the number of miners, and the results are shown in Fig.~\ref{Fig:utility_profit_number}. From Fig.~\ref{Fig:utility_profit_number}, we find that the total service demand of miners and the profit of the CFP increase with the increase of the number of miners in blockchain. This is due to the fact that having more miners will intensify the competition among the miners, which potentially motivates them to have higher service demand. Further, the coming miners have their service demand, and thus the total service demand from miners is increased. In turn, the CFP extracts more surplus from miners and thereby has greater profit gain. Additionally, it is observed that the rate of service demand increment decreases as the number of miners increases. This is from the fact that the incentive of miners to increase their service demand is weakened because the probability of their successful mining is reduced when the number of miners is increasing. Comparing different results, it is also observed that the total service demand of miners and the profit of the CFP increase as $\mu_t$ increases. This is because when $\mu_t$ increases, i.e., the average size of one block becomes larger, the variable reward for each miner also increases. The potential incentive of miners to increase their service demand is improved, and accordingly the total service demand of miners increases. Consequently, the CFP achieves greater profit gain.

\paragraph{The impacts of reward for successful mining}\label{Subsubsubsec:reward}

Then, we investigate the impacts of variable reward and fixed reward on miners and the CFP, which are shown in Fig.~\ref{Fig:utility_profit_reward}. It is observed that with the increase of variable reward factor, both the total service demand of miners and the profit of the CFP increase. This is from the fact that the increased variable reward enhances the motivation of miners for higher service demand, and the total service demand is enhanced accordingly. As a result, the CFP achieves greater profit gain. Further, by comparing curves with different value of fixed reward, we find that as the fixed reward increases, the total service demand of miners and the profit of the CFP also increase. Similarly, this is because the increased fixed reward induces greater incentive of miners, which in turn improves the total service demand of miners and the profit of the CFP.

\paragraph{The impacts of propagation delay}\label{Subsubsubsec:delay}

At last, we examine the impact of propagation delay on miners and the CFP, as illustrated in Fig.~\ref{Fig:utility_profit_delay}. It is observed that as the propagation delay factor increases, the total service demand and the profit of the CFP increase. This is because when the propagation delay effects are strong, the miners with larger mined block need to have higher service demand to reduce the propagation delay of their propagated solutions. At the same time, a miner with smaller mined block is also incentivized from the demand competition with the other miners. Therefore, the total service demand increases, which in turn improves the profit of the CFP. Additionally, we observe that as the value of service cost factor increases, the total service demand decreases under discriminatory pricing and remains unchanged under uniform pricing. On the contrary, the profit of the CFP increases in both schemes. Recall from Fig.~\ref{Fig:price_comparison}, the reason is that the optimal price under uniform pricing remains unchanged from varying the value of service cost factor, and thus the service demand remains unchanged under uniform pricing. Correspondingly, the CFP achieves greater profit gain from the lower cost under uniform pricing. However, under discriminatory pricing, when the service cost decreases, the CFP has an incentive to set lower price for some miners to encourage higher total service demand. On the contrary, when the value of service cost factor increases, the CFP has no incentive to set lower price for these miners, since the higher total service demand results in higher cost for the CFP. Therefore, as the value of service cost factor decreases, the total service demand and the profit of CFP increase.
\subsubsection{Investigation on optimal price under uniform and discriminatory pricing schemes}\label{Subsubsec:price}
Then, to explore the impacts of discriminatory pricing on each specific miner, we investigate the optimal price and resulting individual computing service demand from miners. We conduct a case study for three-miner mining with the following parameters: $t_1=100$, $t_2=200$, $t_3=300$, $\underline x = 10^{-2}$, $\overline x = 100$, $\overline p =100$, $R = 10^4$, $r = 20$, $z = 5 \times 10^{-3}$, and $c=10^{-3}$.

As expected, we observe from Figs.~\ref{Fig:Price_three_miners_R} and~\ref{Fig:Price_three_miners} that the optimal price charging to the miners with the smaller block is lower, e.g., miners $1$ and $2$. This is because the variable reward of miners $1$ and $2$ for successful mining is smaller than that of miner $3$. Thus, the miners $1$ and $2$ have no incentive to pay a high price for their service demand as miner $3$. In this case, the CFP can greatly improve the individual service demand of miners $1$ and $2$ by setting lower prices to attract them, as illustrated in Figs.~\ref{Fig:IndividualDemand_R} and~\ref{Fig:IndividualDemand}. Due to the competition from other two miners, the miner $3$ also has the potential incentive to increase its service demand. However, due to the high service unit price, as a result, the miner $3$ reduces its service demand for saving cost. Nevertheless, the increase of service demand from miners $1$ and $2$ are greater. Therefore, the total service demand and the profit of the CFP are still improved under discriminatory pricing compared with uniform pricing.

\begin{figure}[t]
\centering
\includegraphics[width=0.35\textwidth]{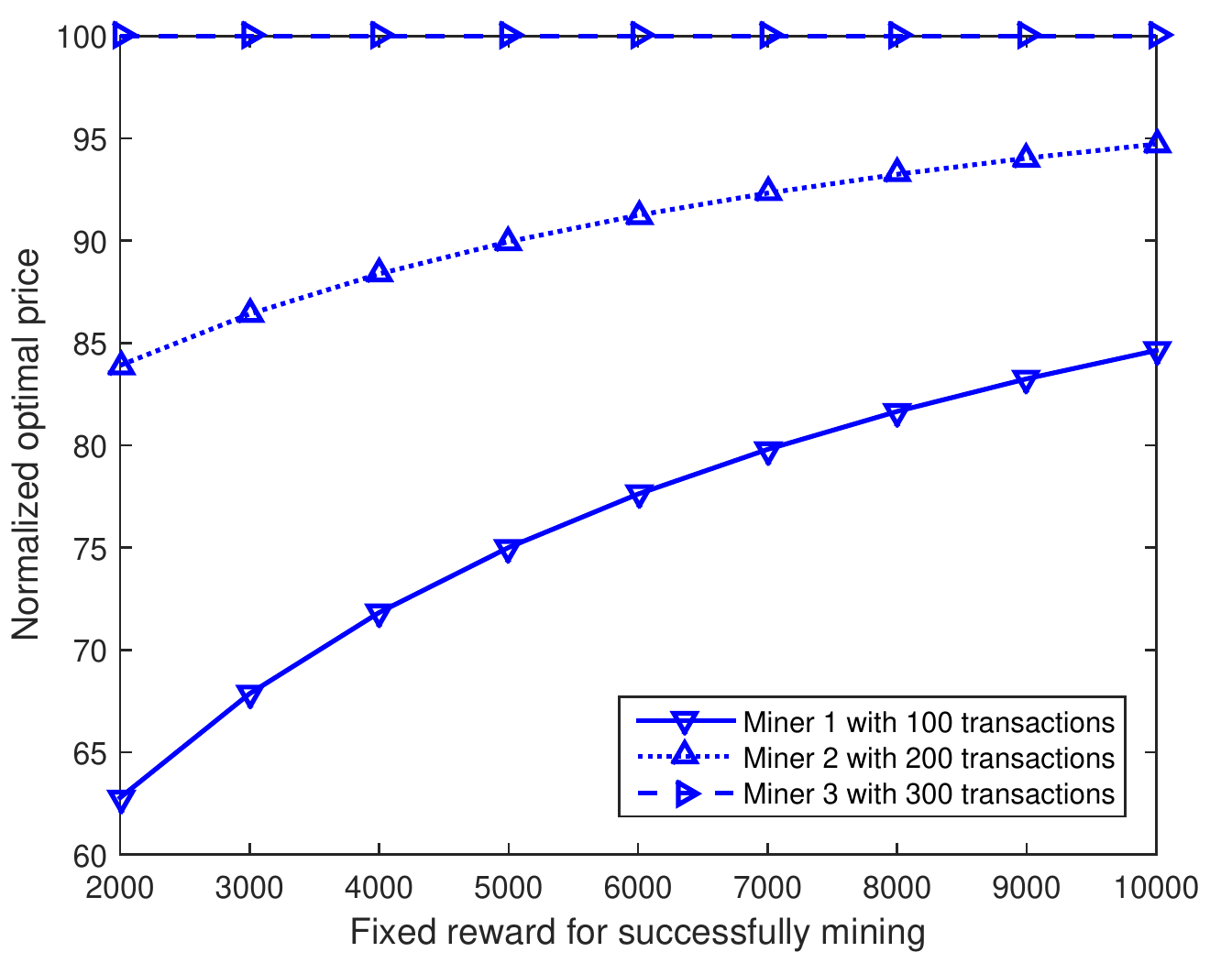}
\caption{Normalized optimal price versus the fixed reward for mining successfully under discriminatory pricing.}\label{Fig:Price_three_miners_R}
\end{figure}

\begin{figure}[t]
\centering
\includegraphics[width=0.35\textwidth]{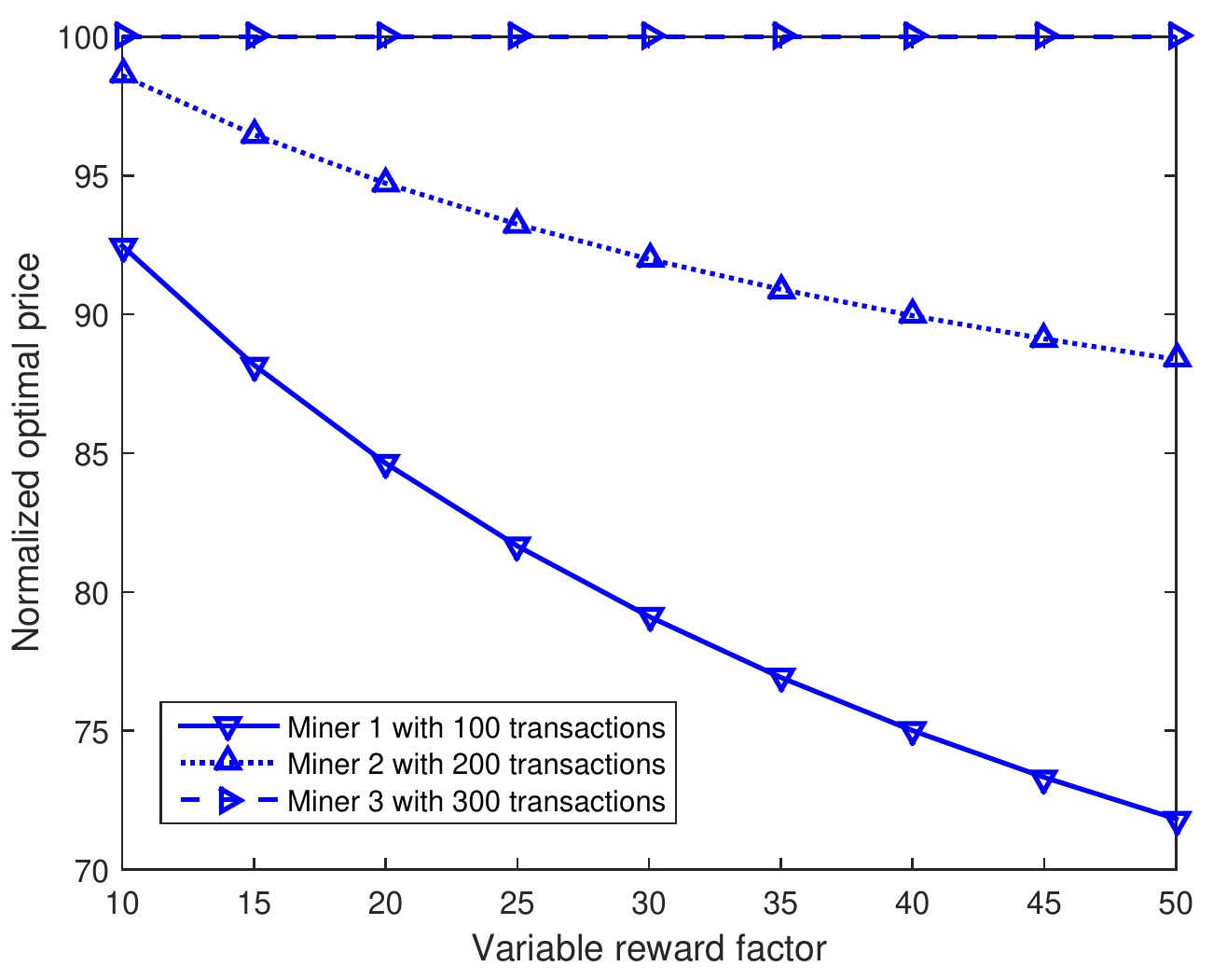}
\caption{Normalized optimal price versus the variable reward factor under discriminatory pricing.}\label{Fig:Price_three_miners}
\end{figure}

\begin{figure}[t]
\centering
\includegraphics[width=0.35\textwidth]{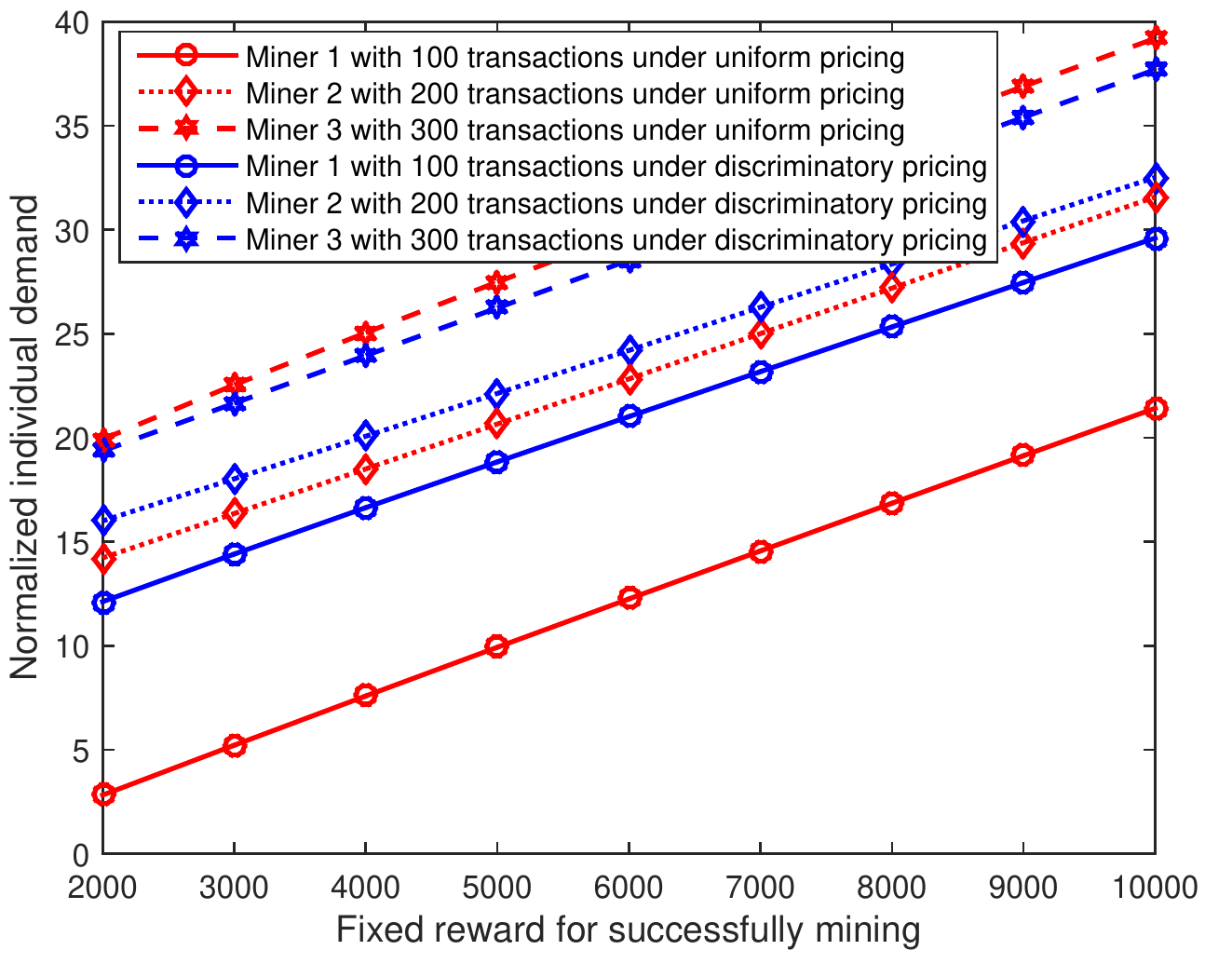}
\caption{Normalized individual demand versus the fixed reward for mining successful.}\label{Fig:IndividualDemand_R}
\end{figure}

\begin{figure}[t]
\centering
\includegraphics[width=0.35\textwidth]{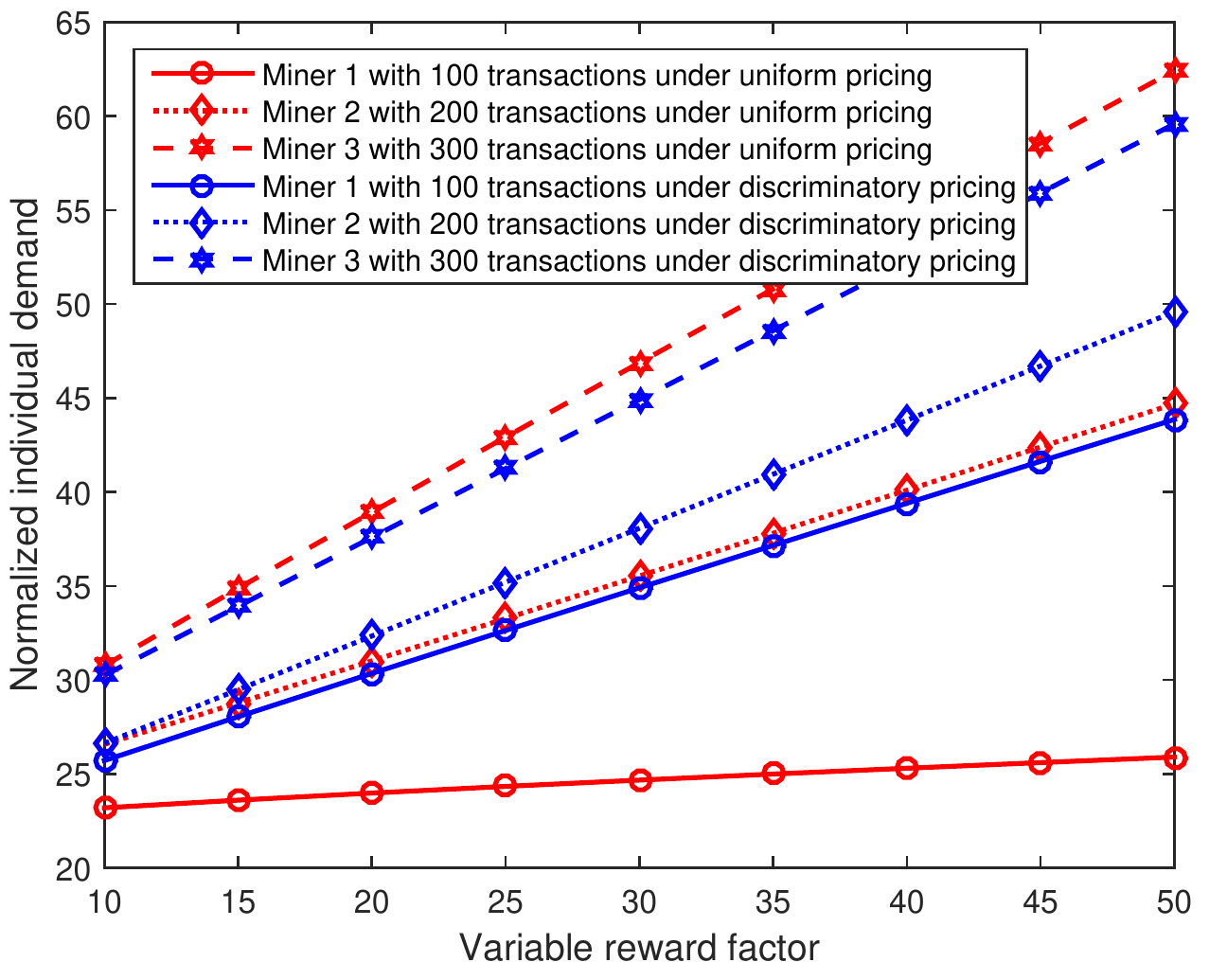}
\caption{Normalized individual demand versus the variable reward factor.}\label{Fig:IndividualDemand}
\end{figure}

Further, from Fig.~\ref{Fig:Price_three_miners_R}, we observe that the optimal prices for miners $1$ and $2$ increase with the increase of fixed reward. This is because as the fixed reward increases, the incentives of miners $1$ and $2$ to have higher service demand is greater. In this case, the CFP is able to raise the price and charge more for higher revenue, and thus achieves greater profit. Therefore, for each miner, the individual service demand increases as the fixed reward increases, as shown in Fig.~\ref{Fig:IndividualDemand_R}. Additionally, we observe from Fig.~\ref{Fig:Price_three_miners} that the optimal prices for miners $1$ and $2$ decrease as the variable reward factor increases. This is because when the variable reward factor increases, the incentive of each miner to have higher service demand is greater. However, the incentives of the miners with smaller block to mine, i.e., the miners $1$ and $2$ are still not much as that of miner $3$, and become smaller than that of miner $3$ as the variable reward factor increases. Therefore, the CFP intends to set the lower price for miners $1$ and $2$ which may induce more individual service demand as shown in Fig.~\ref{Fig:IndividualDemand}.

Note that the Stackelberg game of the edge/fog computing service for blockchain aims
at maximizing the profit of the CFP. Alternatively, social welfare, i.e., utility of miners, are
also important and should be maximized. As such, auction~\cite{jiang2015data} is a suitable tool to achieve this objective in which some preliminary modeling and results are presented in~\cite{jiao2017blockchain}.

\section{Conclusion}\label{Sec:Conclusion}

In this paper, we have investigated the price-based computing resource management, for supporting offloading mining tasks to cloud/fog provider in proof-of-work based public blockchain networks. In particular, we have adopted the two-stage Stackelberg game model to jointly study the profit maximization of cloud/fog provider and the utility maximization of miners. Through backward induction, we have derived the unique Nash equilibrium point of the game among the miners. The optimal resource management schemes including the uniform and discriminatory pricing for the cloud/fog provider have been presented and examined. Further, the existence and uniqueness of the Stackelberg equilibrium have been proved analytically for both pricing schemes. We have performed the real experiment to validate the proposed analytical model. Additionally, we have conducted the numerical simulations to evaluate the network performance, which help the cloud/fog provider to achieve optimal resource management and gain the highest profit. For the future work, we will further study the oligopoly market with multiple cloud/fog providers, where providers compete with each other for selling computing services to miners. Another direction is to study the optimal strategies of the provider and miners with the consideration of cyber-attacks, such as~\cite{feng2018cyber}.

\section*{Acknowledgement}
This work was supported in part by WASP/NTU M4082187 (4080), Singapore MOE Tier 1 under Grant 2017-T1-002-007 RG122/17, MOE Tier 2 under Grant MOE2014-T2-2-015 ARC4/15, NRF2015-NRF-ISF001-2277, EMA Energy Resilience under Grant NRF2017EWT-EP003-041, and in part by US MURI, NSF CNS-1717454, CNS- 1731424, CNS-1702850, CNS-1646607,and ECCS-1547201.

\bibliography{bibfile}

\begin{thebibliography}{10}

\bibitem{xiong2017optimal}
Z.~Xiong, S.~Feng, D.~Niyato, P.~Wang  and Z.~Han,
\newblock ``Optimal pricing-based edge computing resource management in mobile
  blockchain,''
\newblock in {\em Proceedings of IEEE ICC}, Kansas City, MO, May 2018.

\bibitem{Bitcoin}
S.~Nakamoto,
\newblock ``Bitcoin: A peer-to-peer electronic cash system,''
\newblock {\em Self-published Paper}, May 2008.

\bibitem{wang2018survey}
W.~Wang, D.~T. Hoang, Z.~Xiong, D.~Niyato, P.~Wang, P.~Hu  and Y.~Wen,
\newblock ``A survey on consensus mechanisms and mining management in
  blockchain networks,''
\newblock {\em arXiv preprint arXiv:1805.02707}, 2018.

\bibitem{castro1999practical}
M.~Castro, B.~Liskov  et~al.,
\newblock ``Practical byzantine fault tolerance,''
\newblock in {\em Proceedings of the third symposium on Operating systems
  design and implementation}, New Orleans, LA, Feb. 1999, vol.~99, pp.
  173--186.

\bibitem{ongaro2014search}
D.~Ongaro and J.~K. Ousterhout,
\newblock ``In search of an understandable consensus algorithm.,''
\newblock in {\em USENIX Annual Technical Conference}, Philadelphia, PA, June
  2014, pp. 305--319.

\bibitem{Vukolic2016}
M.~Vukoli{\'{c}},
\newblock ``The quest for scalable blockchain fabric: Proof-of-work vs. bft
  replication,''
\newblock in {\em Open Problems in Network Security: IFIP WG 11.4 International
  Workshop}, Zurich, Switzerland, Oct. 2015, pp. 112--125.

\bibitem{Garay2015}
J.~Garay, A.~Kiayias  and N.~Leonardos,
\newblock ``The bitcoin backbone protocol: Analysis and applications,''
\newblock in {\em Advances in Cryptology - EUROCRYPT 2015: 34th Annual
  International Conference on the Theory and Applications of Cryptographic
  Techniques, Part II}, Sofia, Bulgaria, Apr. 2015, pp. 281--310.

\bibitem{maesa2017blockchain}
D.~D.~F. Maesa, P.~Mori  and L.~Ricci,
\newblock ``Blockchain based access control,''
\newblock in {\em IFIP International Conference on Distributed Applications and
  Interoperable Systems}, Neuchatel, Switzerland, June 2017.

\bibitem{Wang1805:Decentralized}
W.~Wang, D.~Niyato, P.~Wang  and A.~Leshem,
\newblock ``Decentralized caching for content delivery based on blockchain: A
  game theoretic perspective,''
\newblock in {\em Proceedings of IEEE ICC}, Kansas City, MO, May 2018.

\bibitem{chen2017framework}
X.~Chen, S.~Chen, X.~Zeng, X.~Zheng, Y.~Zhang  and C.~Rong,
\newblock ``Framework for context-aware computation offloading in mobile cloud
  computing,''
\newblock {\em Journal of Cloud Computing}, vol. 6, no. 1, pp. 1, 2017.

\bibitem{huang2017vehicular}
C.~Huang, R.~Lu  and K.-K.~R. Choo,
\newblock ``Vehicular fog computing: architecture, use case, and security and
  forensic challenges,''
\newblock {\em IEEE Communications Magazine}, vol. 55, no. 11, pp. 105--111,
  2017.

\bibitem{recabarren2017hardening}
R.~Recabarren and B.~Carbunar,
\newblock ``Hardening stratum, the bitcoin pool mining protocol,''
\newblock {\em arXiv preprint arXiv:1703.06545}, 2017.

\bibitem{zhang2017hierarchical}
H.~Zhang, Y.~Zhang, Y.~Gu, D.~Niyato  and Z.~Han,
\newblock ``A hierarchical game framework for resource management in fog
  computing,''
\newblock {\em IEEE Communications Magazine}, vol. 55, no. 8, pp. 52--57, 2017.

\bibitem{zhang2017multi}
H.~Zhang, Y.~Xiao, L.~X. Cai, D.~Niyato, L.~Song  and Z.~Han,
\newblock ``A multi-leader multi-follower stackelberg game for resource
  management in lte unlicensed,''
\newblock {\em IEEE Transactions on Wireless Communications}, vol. 16, no. 1,
  pp. 348--361, 2017.

\bibitem{jiang2014optimal}
C.~Jiang, Y.~Chen, K.~R. Liu  and Y.~Ren,
\newblock ``Optimal pricing strategy for operators in cognitive femtocell
  networks,''
\newblock {\em IEEE Transactions on Wireless Communications}, vol. 13, no. 9,
  pp. 5288--5301, 2014.

\bibitem{laffont1998network}
J.-J. Laffont, P.~Rey  and J.~Tirole,
\newblock ``Network competition: {II}. price discrimination,''
\newblock {\em The RAND Journal of Economics}, pp. 38--56, 1998.

\bibitem{dinh2017untangling}
T.~T.~A. Dinh, R.~Liu, M.~Zhang, G.~Chen, B.~C. Ooi  and J.~Wang,
\newblock ``Untangling blockchain: A data processing view of blockchain
  systems,''
\newblock {\em arXiv preprint arXiv:1708.05665}, 2017.

\bibitem{tschorsch2016bitcoin}
F.~Tschorsch and B.~Scheuermann,
\newblock ``Bitcoin and beyond: A technical survey on decentralized digital
  currencies,''
\newblock {\em IEEE Communications Surveys \& Tutorials}, vol. 18, no. 3, pp.
  2084--2123, 2016.

\bibitem{8024034}
D.~Chatzopoulos, M.~Ahmadi, S.~Kosta  and P.~Hui,
\newblock ``Flopcoin: A cryptocurrency for computation offloading,''
\newblock {\em IEEE Transactions on Mobile Computing}, vol. PP, no. 99, pp.
  1--1, 2017.

\bibitem{7966965}
H.~Kopp, D.~Mödinger, F.~Hauck, F.~Kargl  and C.~Bösch,
\newblock ``Design of a privacy-preserving decentralized file storage with
  financial incentives,''
\newblock in {\em 2017 IEEE European Symposium on Security and Privacy
  Workshops (EuroS PW)}, Paris, France, 2017.

\bibitem{8260929}
J.~Backman, S.~Yrjölä, K.~Valtanen  and O.~Mämmelä,
\newblock ``Blockchain network slice broker in 5g: Slice leasing in factory of
  the future use case,''
\newblock in {\em 2017 Internet of Things Business Models, Users, and
  Networks}, Copenhagen, Denmark, Nov. 2017, pp. 1--8.

\bibitem{kang2017enabling}
J.~Kang, R.~Yu, X.~Huang, S.~Maharjan, Y.~Zhang  and E.~Hossain,
\newblock ``Enabling localized peer-to-peer electricity trading among plug-in
  hybrid electric vehicles using consortium blockchains,''
\newblock {\em IEEE Transactions on Industrial Informatics}, vol. 13, no. 6,
  pp. 3154--3164, Dec 2017.

\bibitem{yang2012crowdsourcing}
D.~Yang, G.~Xue, X.~Fang  and J.~Tang,
\newblock ``Crowdsourcing to smartphones: incentive mechanism design for mobile
  phone sensing,''
\newblock in {\em Proceedings of the 18th annual international conference on
  Mobile computing and networking}. ACM, 2012, pp. 173--184.

\bibitem{chakeri2017incentive}
A.~Chakeri and L.~Jaimes,
\newblock ``An incentive mechanism for crowdsensing markets with multiple
  crowdsourcers,''
\newblock {\em IEEE Internet of Things Journal}, vol. 5, pp. 708--715, 2018.

\bibitem{chakeri2017iterative}
A.~Chakeri and L.~Jaimes,
\newblock ``An iterative incentive mechanism design for crowd sensing using
  best response dynamics,''
\newblock in {\em Proceedings of IEEE ICC}, Paris, France, May 2017.

\bibitem{houy2014bitcoin}
N.~Houy,
\newblock ``The bitcoin mining game,''
\newblock {\em Ledger Journal}, vol. 1, no. 13, pp. 53 -- 68, 2016.

\bibitem{beccuti2017bitcoin}
J.~Beccuti and C.~Jaag,
\newblock ``{The Bitcoin Mining Game: On the Optimality of Honesty in
  Proof-of-work Consensus Mechanism},''
\newblock Working Papers 0060, Swiss Economics, Aug. 2017.

\bibitem{kiayias2016blockchain}
A.~Kiayias, E.~Koutsoupias, M.~Kyropoulou  and Y.~Tselekounis,
\newblock ``Blockchain mining games,''
\newblock in {\em Proceedings of the ACM Conference on Economics and
  Computation (EC)}, Maastricht, Netherlands, July 2016.

\bibitem{lewenberg2015bitcoin}
Y.~Lewenberg, Y.~Bachrach, Y.~Sompolinsky, A.~Zohar  and J.~S. Rosenschein,
\newblock ``Bitcoin mining pools: A cooperative game theoretic analysis,''
\newblock in {\em Proceedings of the ACM AAMAS}, Istanbul, Turkey, May 2015.

\bibitem{fisch2017socially}
B.~A. Fisch, R.~Pass  and A.~Shelat,
\newblock ``Socially optimal mining pools,''
\newblock {\em arXiv preprint arXiv:1703.03846}, 2017.

\bibitem{kim2016group}
S.~Kim,
\newblock ``Group bargaining based bitcoin mining scheme using incentive
  payment process,''
\newblock {\em Transactions on Emerging Telecommunications Technologies}, vol.
  27, no. 11, pp. 1486--1495, 2016.

\bibitem{luu2015power}
L.~Luu, R.~Saha, I.~Parameshwaran, P.~Saxena  and A.~Hobor,
\newblock ``On power splitting games in distributed computation: The case of
  bitcoin pooled mining,''
\newblock in {\em Proceedings of IEEE CSF}, Verona, Italy, July 2015.

\bibitem{xiong2018mobile}
Z.~Xiong, Y.~Zhang, D.~Niyato, P.~Wang  and Z.~Han,
\newblock ``When mobile blockchain meets edge computing,''
\newblock {\em IEEE Communications Magazine}, vol. 56, pp. 33--39, August,
  2018.

\bibitem{luong2017blockchain}
N.~C. Luong, Z.~Xiong, P.~Wang  and D.~Niyato,
\newblock ``Optimal auction for edge computing resource management in mobile
  blockchain networks: A deep learning approach,''
\newblock in {\em Proceedings of IEEE ICC}, Kansas City, MO, May 2018.

\bibitem{jiao2017blockchain}
Y.~Jiao, P.~Wang, D.~Niyato  and Z.~Xiong,
\newblock ``Social welfare maximization auction in edge computing resource
  allocation for mobile blockchain,''
\newblock in {\em Proceedings of IEEE ICC}, Kansas City, MO, May 2018.

\bibitem{wood2014ethereum}
G.~Wood,
\newblock ``Ethereum: A secure decentralised generalised transaction ledger
  (eip-150 revision),''
\newblock {\em Ethereum Project Yellow Paper}, 2017.

\bibitem{tosh2017security}
D.~K. Tosh, S.~Shetty, X.~Liang, C.~A. Kamhoua, K.~A. Kwiat  and L.~Njilla,
\newblock ``Security implications of blockchain cloud with analysis of block
  withholding attack,''
\newblock in {\em Proceedings of IEEE/ACM CCGrid}, 2017, pp. 458--467.

\bibitem{NanTSC}
N.~Wang, B.~Varghese, M.~Matthaiou  and D.~S. Nikolopoulos,
\newblock ``Enorm: A framework for edge node resource management,''
\newblock {\em IEEE Transactions on Services Computing}, vol. PP, pp. 1--1,
  2017.

\bibitem{Approximation}
``Orphan probablity approximation,''
  https://gist.github.com/gavinandres\\en/5044482.

\bibitem{senmarti2015analysis}
E.~Senmarti~Robla,
\newblock {\em Analysis of Reward Strategy and Transaction Selection in Bitcoin
  Block Generation},
\newblock Ph.D. thesis, University of Washington, 2015.

\bibitem{decker2013information}
C.~Decker and R.~Wattenhofer,
\newblock ``Information propagation in the bitcoin network,''
\newblock in {\em Proceedings of IEEE P2P}, Trento, Italy, September 2013.

\bibitem{han2012game}
Z.~Han, D.~Niyato, W.~Saad, T.~Baar  and A.~Hjrungnes,
\newblock {\em Game theory in wireless and communication networks: theory,
  models, and applications},
\newblock Cambridge University Press, 2012.

\bibitem{scutari2010convex}
G.~Scutari, D.~P. Palomar, F.~Facchinei  and J.-s. Pang,
\newblock ``Convex optimization, game theory, and variational inequality
  theory,''
\newblock {\em IEEE Signal Processing Magazine}, vol. 27, no. 3, pp. 35--49,
  2010.

\bibitem{suankaewmanee2018performance}
K.~Suankaewmanee, D.~T. Hoang, D.~Niyato, S.~Sawadsitang, P.~Wang  and Z.~Han,
\newblock ``Performance analysis and application of mobile blockchain,''
\newblock in {\em 2018 International Conference on Computing, Networking and
  Communications (ICNC)}. IEEE, 2018, pp. 642--646.

\bibitem{Docker}
``Docker,'' https://www.docker.com/community-edition.

\bibitem{Ethereum}
``Go-ethereum,'' https://ethereum.github.io/go-ethereum/.

\bibitem{jiang2015data}
C.~Jiang, Y.~Chen, Q.~Wang  and K.~R. Liu,
\newblock ``Data-driven auction mechanism design in iaas cloud computing,''
\newblock {\em IEEE Transactions on Services Computing}, early access, 2018.

\bibitem{feng2018cyber}
S.~Feng, W.~Wang, Z.~Xiong, D.~Niyato, P.~Wang  and S.~S. Wang,
\newblock ``On cyber risk management of blockchain networks: A game theoretic
  approach,''
\newblock {\em arXiv preprint arXiv:1804.10412}, 2018.

\end{thebibliography}

\end{document}